\documentclass[prd,aps,eqsecnum,superscriptaddress,nofootinbib,notitlepage,longbibliography,draft]{revtex4-2}
\usepackage[T1]{fontenc}
\usepackage{lmodern}
\usepackage{mathtools,amssymb,amsthm}
\usepackage{microtype}
\usepackage{comment}
\usepackage{xcolor}
\usepackage[final,hidelinks,pagebackref]{hyperref}
\backrefparscanfalse
\renewcommand*{\backref}[1]{}
\renewcommand*{\backrefalt}[4]{%
  \ifcase #1\relax
  \or\space(Cited on page~\textcolor{blue}{#2})%
  \else\space(Cited on pages~\textcolor{blue}{#2})%
  \fi}
\hypersetup{pdftitle={Poisson identities for relative-entropy loss and quantum recovery},
  pdfsubject={Exact entropy-loss identities; pointwise Petz recovery and general-channel modular matrix differences}}
\numberwithin{equation}{section}
\newtheorem{theorem}{Theorem}[section]
\newtheorem{proposition}[theorem]{Proposition}
\newtheorem{lemma}[theorem]{Lemma}
\newcommand{\restatedlemmanumber}{}
\newtheorem*{restatedlemma}{Lemma \restatedlemmanumber}
\newtheorem{corollary}[theorem]{Corollary}
\theoremstyle{remark}
\newtheorem{remark}[theorem]{Remark}
\newcommand{\R}{\mathbb R}
\newcommand{\C}{\mathbb C}
\newcommand{\HH}{\mathcal H}
\newcommand{\BB}{\mathcal B}
\newcommand{\one}{\mathbf 1}

\newcommand{\Tr}{\operatorname{Tr}}
\newcommand{\supp}{\operatorname{supp}}
\newcommand{\DM}{D_{\mathrm M}}
\newcommand{\HS}[1]{\left\|#1\right\|_2}
\newcommand{\op}[1]{\left\|#1\right\|_\infty}

\newcommand{\Pcal}{\mathcal P}

\newcommand{\dd}{\,\mathrm d}
\allowdisplaybreaks[2]
\RequirePackage{mathrsfs}

\begin{document}
\title[Poisson identities and quantum recovery]{Quantum data processing equality}
\author{Tai-Hsuan Yang}
\email{ty39@illinois.edu}
\affiliation{Department of Physics and Institute for Condensed Matter Theory,
University of Illinois Urbana-Champaign}
\date{\today}
\begin{abstract}
The quantum data processing inequality asserts that quantum relative entropy is monotonically 
non-increasing under quantum channels. Here, we promote this fundamental inequality to an 
exact equality via boundary and surface Poisson integral representations across a complex strip. 
The relative-entropy loss is resolved into two independent, nonnegative physical mechanisms: 
an operational Petz recovery component governing state reconstructibility, and a boundary term 
measuring the mismatch of relative modular dynamics. We further express the loss as the
measured relative entropy between the original state and its averaged Petz recovery state plus
a nonnegative variational remainder. For the Belavkin--Staszewski (BS) loss, analogous
integral representations yield a decomposition into negative log Uhlmann fidelity with
an averaged reconstruction plus a nonnegative variational remainder.
Furthermore, we establish recovery bounds at any prescribed rotation parameter,
including the canonical unrotated Petz map, with optimal square-root scaling
in the relative entropy loss. The prefactors depend on distinct-eigenvalue counts or
logarithmic spectral quantities, improving upon the inverse-power dependence
on small eigenvalues in previous bounds.
\end{abstract}
\maketitle

\section*{AI usage statement}\label{sec:ai-usage}
The results presented in this work were obtained using GPT-5.6 Sol Pro and
GPT-6 Astra Pro. The author takes full responsibility for the content of
this manuscript, including the correctness of all results, proofs, and
conclusions. The initial result for the integral form of the CMI arose during the preparation
of~\cite{yang2026modularcommutatorrobusttopological} because we wanted to show that small conditional mutual information
implies a small Markov-decomposition error for the modular flow. Later we realized the significance of this result and generalized it to the relative-entropy loss under quantum channels.
The Poisson integral representation of CMI was first discovered by
GPT-5.6 Sol Pro in
\href{https://chatgpt.com/share/6aab1a12-e410-83ea-9c23-4ea856943f1d}{this conversation}.
We subsequently obtained its generalization to relative-entropy loss
under quantum channels in
\href{https://chatgpt.com/share/6ab06667-fd64-83e9-bcc5-a4405b711d95}{this subsequent conversation}.
Some proposed integral formulas for other divergences not discussed in this paper generated by GPT-6 Astra Pro are recorded in
\href{https://chatgpt.com/share/6ab06810-9020-83ea-9029-ea52f513377c}{this further conversation}.

\tableofcontents
\section{Introduction}\label{sec:main}

Monotonicity of the quantum relative entropy under completely positive, trace-preserving 
(CPTP) maps is among the most foundational principles of quantum information theory~\cite{Lindblad75,lieb1973proof}. 
For density matrices $\rho, \sigma$ and a quantum channel $\Phi$, the quantum data processing 
inequality asserts that the relative-entropy loss,
\begin{equation}
  \Delta(\rho,\sigma;\Phi) := D(\rho\|\sigma) - D(\Phi(\rho)\|\Phi(\sigma)),
\end{equation}
is nonnegative. The landmark theorem of Petz~\cite{Petz86,Petz88} characterized the exact equality 
case: $\Delta = 0$ if and only if there exists a common recovery channel reversing the action 
of $\Phi$ on both states. This recovery channel admits an explicit construction, known as the Petz recovery map.

Over the past decade, a major research program has focused on establishing \emph{quantitative} 
strengthenings of this inequality, relating a small entropy loss $\Delta \ll 1$ to approximate 
state reconstructibility. Initiated by Fawzi and Renner~\cite{Fawzi2014} for the conditional mutual 
information (CMI) and subsequently generalized to arbitrary quantum channels and von Neumann 
algebras~\cite{SFR,Wilde,Junge2015,SutterBertaTomamichel2017,JL}, it was shown that $\Delta$ 
controls state reconstruction under an averaged family of rotated Petz maps $\mathcal{R}_s$ ($s \in \mathbb{R}$). 
These recovery bounds take several forms. Berta, Lemm, and
Wilde~\cite{BertaLemmWilde2015} showed that suitable unitary rotations of the
Petz map yield a recovery fidelity controlled by the entropy loss.
Junge et al.~\cite{Junge2015} subsequently constructed an explicit average
of rotated Petz maps, depending only on the reference state and the channel,
and established an averaged negative-log-fidelity bound.
Junge and LaRacuente~\cite{JL} further developed $p$-fidelity recovery bounds,
including a squared $2$-norm estimate for recovery of square-root densities.
Bounds in terms of measured relative entropy arise from
state redistribution for CMI~\cite{BrandaoHarrowOppenheimStrelchuk2015}
and from pinching for general channels~\cite{pinchedpetz2016}.
Multivariate trace inequalities~\cite{SutterBertaTomamichel2017,Sutter2018} further give
a lower bound on the entropy loss in terms of $\DM(\rho\|\tau)$, where
$\tau$ is an explicit average of rotated Petz recovered states.

Despite these advances, existing recovery frameworks leave several fundamental questions open:
\begin{enumerate}
  \item \textbf{One-sided bounds versus exact equality:} 
  It has remained unknown whether the quantum data processing inequality can be promoted to an 
  \emph{exact equality}, and what physical and mathematical mechanism governs the discarded remainder.
  \item \textbf{Pointwise versus averaged recovery:} While universal recovery channels average 
  over the rotation parameter $s$, operational applications in quantum error correction and many-body 
  physics often require reconstruction at a \emph{single, prescribed} rotation parameter—most 
  notably the canonical unrotated Petz map at $s = 0$. Existing pointwise estimates~\cite{CV,GaoWilde} 
  either suffer from sub-optimal scaling ($O(\Delta^{1/4})$ or $O(\sqrt{\Delta \log(1/\Delta)})$) or 
  introduce singular dependencies on minimum eigenvalues, such as $Q = \Tr(\rho^2 \sigma^{-1})$, 
  which blow up for low-entropy or nearly orthogonal states.
\end{enumerate}

In this paper, we address these questions by deriving an exact Poisson
integral representation of relative-entropy loss and using it to establish
Petz recovery bounds at any prescribed rotation parameter, including $s=0$.
For faithful (full-rank) finite-dimensional states, these bounds have the optimal
$O(\sqrt{\Delta})$ trace-norm scaling, with coefficients controlled by
distinct-eigenvalue counts or logarithmic spectral quantities.
Our proof of the identity shares the complex-analytic strip
framework used by Junge et al.~\cite{Junge2015}. Their argument applies
Hirschman's strengthened three-line theorem and H\"older's inequality to
obtain a recovery bound. Here, we apply the Poisson integral formula
directly to the real part of a scalar analytic function and differentiate
at the boundary, obtaining an exact identity that retains both boundary
contributions. 

Across the complex strip
$0 \le \operatorname{Re} z \le 1/2$, the relative-entropy loss decomposes cleanly into two nonnegative boundary contributions.
\begin{equation}\label{eq:intro-loss-schematic}
  \Delta = \Delta_{\mathrm{mod}} + \Delta_{\mathrm{rec}}.
\end{equation}
The two terms possess distinct physical interpretations:
\begin{itemize}
 \item \textbf{Modular mismatch ($\Delta_{\mathrm{mod}}$):}
 The lower boundary measures how the channel fails to preserve the
 relative modular evolution $\sigma^{is}\rho^{-is}$ of the two states.

 \item \textbf{Petz recovery ($\Delta_{\mathrm{rec}}$):}
 The upper boundary controls the reconstruction error of rotated
 Petz recovery, averaged over the rotation parameter.
\end{itemize}
Standard recovery inequalities emerge simply by discarding $\Delta_{\mathrm{mod}} \ge 0$. 
Our identity identifies this term as a nonnegative contribution to the entropy loss
that quantifies the mismatch of relative modular dynamics under the channel.

\subsection*{Summary of Main Contributions}
\begin{itemize}
  \item \textbf{Exact Relative-Entropy Loss and Measured Recovery (Theorems~\ref{thm:loss} and~\ref{thm:measured}):} 
  We establish the exact boundary Poisson representation for $\Delta$. By tilting the underlying analytic 
  function with a positive test operator $\omega$, we obtain an exact variational formula for the gap in the
  measured-relative-entropy recovery bound,
  \begin{equation}\label{eq:intro-measured-remainder}
   \Delta-\DM(\rho\|\tau)=\inf_{\omega>0}\mathcal P(\omega)\ge0,
  \end{equation}
  completing the known inequality of Sutter, Berta, and Tomamichel~\cite{SutterBertaTomamichel2017} into an equality.

  \item \textbf{Belavkin--Staszewski Divergence and Fidelity Remainder (Theorems~\ref{thm:bs-poisson} and~\ref{thm:bs-fidelity-remainder}):} 
  We derive parallel boundary and surface identities for the loss of the
  Belavkin--Staszewski (BS) relative entropy~\cite{BelavkinStaszewski1982},
  whose strengthened data-processing inequalities were studied
  in~\cite{BluhmCapel2020}. We obtain the exact fidelity remainder
  \begin{equation}\label{eq:intro-bs-remainder}
   \Delta_{\mathrm{BS}}+\log F(\rho,\overline\tau_{\mathrm{BS}})
   =\inf_{\omega>0}\mathcal P_{\mathrm{BS}}(\omega)\ge0,
  \end{equation}
  where $F(a,b)=\|a^{1/2}b^{1/2}\|_1^2$ and
  $\overline\tau_{\mathrm{BS}}$ is a positive reconstruction operator
  determined by $\sigma$, $\Phi$, and the output $\Phi(\rho)$.
  It is obtained by averaging over the rotation parameter as in
  \eqref{eq:bs-averaged-reconstruction} and has trace at most one.

  \item \textbf{Optimal Pointwise Recovery without Singular Spectral Bounds (Theorem~\ref{thm:finite-recovery}):} 
  Treating the boundary errors as finite exponential polynomials, we prove sharp Cauchy-type interpolation 
  lemmas that yield pointwise bounds at \emph{any prescribed} rotation $s_0 \in \mathbb{R}$ (including canonical 
  Petz $s_0 = 0$) with strictly optimal $O(\sqrt{\Delta})$ scaling in trace distance. The prefactors depend 
  only linearly on distinct-eigenvalue counts or logarithmically on operator norms, improving upon previous 
  inverse-eigenvalue dependencies such as $\Tr(\rho^2 \sigma^{-1})$.

  \item \textbf{Surface Representation and Interior Bound (Theorems~\ref{sv:U-surface-thm} and~\ref{sv:complex-interior}):}
  Green's identity expresses $\Delta$ as a weighted surface integral of
  $\HS{\Gamma'}^2$ plus nonnegative boundary correction terms.
  The full boundary numerator also has a nonnegative harmonic extension,
  giving an interior pointwise bound whose coefficient depends only on
  the complex parameter, independently of dimensions and spectra.

  \item \textbf{Modular Flow Factorization for Quantum Markov Chains (Corollary~\ref{cor:CMI}):} 
  For conditional mutual information (where correction terms vanish identically), the identity partitions 
  $I(X:Z\mid Y)_\rho$ into Petz state-reconstruction error plus modular-flow factorization error. This bridges 
  the operational view of quantum Markov states with the dynamical locality of modular flows in topological 
  phases~\cite{Kim2022c-minus,Modular-commutator-Gapped,yang2026modularcommutatorrobusttopological}.
\end{itemize}

\subsection*{Organization of the Paper}
Section~\ref{sec:results} formally presents the mathematical setup and states the Poisson boundary and 
strip surface identities for Umegaki relative entropy, measured recovery, CMI, and the BS divergence. 
Section~\ref{sec:channel} provides the proofs of the boundary Poisson formula, the gap in the measured-relative-entropy recovery bound, 
and the vector Green surface identity. 
Section~\ref{sv:pointwise} establishes the finite-frequency interpolation lemmas and proves the pointwise 
recovery bounds at prescribed parameters.

\section{Main results}\label{sec:results}
\subsection{Setup and notation}
Let
\[
 \Phi:\BB(\HH_A)\longrightarrow\BB(\HH_B)
\]
be completely positive and trace preserving, namely a quantum channel. Denote its Hilbert--Schmidt adjoint
by $\Phi^\dagger$. Let $\rho,\sigma>0$ be density matrices on
$\HH_A$. We use $D(a\|b)=\Tr a(\log a-\log b)$ for relative entropy and all logarithms are natural.
We write $\|X\|_p$ for the Schatten $p$-norm, with $\|X\|_\infty$ denoting the operator norm.
For density matrices $a,b$, let $\DM(a\|b)$ denote the supremum of the
classical relative entropy of their outcome distributions over all finite
positive operator-valued measurements. All complex powers are defined by $a^z=\exp(z\log a)$.
The rotated Petz channels and averaged recovered state are
\begin{equation}\label{eq:recovery-map}
 \mathcal R_s(X) =\sigma^{1/2+is}\Phi^\dagger\!\left( \Phi(\sigma)^{-1/2-is}X\Phi(\sigma)^{-1/2+is} \right)\sigma^{1/2-is},\qquad \tau=\int_\R\frac{\pi}{\cosh(2\pi s)+1}\mathcal R_s(\Phi(\rho))\dd s.
\end{equation}
Each $\mathcal R_s$ is trace preserving and satisfies $\mathcal R_s(\Phi(\sigma))=\sigma$. Now we define the function that plays a central role in this paper:
\begin{equation}\label{eq:Gamma}
 \Gamma(z)
 =\sigma^z
   \Phi^\dagger\!\left(\Phi(\sigma)^{-z}\Phi(\rho)^z\right)
   \rho^{1/2-z},\qquad 0\leq\Re z\leq\frac12.
\end{equation}
This is a bounded analytic matrix-valued function on the strip.
And for a positive definite operator $\omega$ on $\HH_A$, define
\begin{equation}\label{eq:defect-norms}
 \varepsilon_0(s)=1-\HS{\Gamma(-is)}^2,
 \qquad
 \varepsilon_{1,\omega}(s)
   =\Tr(\omega\mathcal R_s(\Phi(\rho)))-\HS{\omega^{1/2+is}\Gamma(1/2+is)}^2,
\end{equation}
and write $\varepsilon_1=\varepsilon_{1,\one}$.
The nonnegativity of those quantities and their relation to the Petz recovered state $\tau$ are shown in
Appendix~\ref{app:boundary-nonnegativity}.
For convenience, we denote
\begin{equation}\label{sv:Gamma-repeat}
 f(z)=\Tr\left(\rho^{1/2}\Gamma(z)\right),\qquad
 \Delta=D(\rho\|\sigma)-D(\Phi(\rho)\|\Phi(\sigma)).
\end{equation}
Here $0<x<1/2$ and $z=x+it$.

\subsection{Entropy-loss and measured-recovery identities}\label{sv:surface}
\begin{theorem}[Poisson strip and surface integrals for relative entropy loss]\label{thm:loss}\label{sv:U-surface-thm}
Under the setup,
\begin{equation}\label{eq:loss-expanded}
 \Delta =\pi\int_\R\left[
 \frac{\HS{\rho^{1/2}-\Gamma(-is)}^2+\varepsilon_0(s)}{\cosh(2\pi s)-1}
 +\frac{\HS{\rho^{1/2}-\Gamma(1/2+is)}^2+\varepsilon_1(s)}{\cosh(2\pi s)+1}
 \right]\dd s.
\end{equation}

By unitary invariance of the Hilbert--Schmidt norm,
\[
 \HS{\rho^{1/2}-\Gamma(-is)}
 =\HS{\left[\sigma^{is}\rho^{-is}
 -\Phi^\dagger\!\left(\Phi(\sigma)^{is}\Phi(\rho)^{-is}\right)
 \right]\rho^{1/2}}.
\]
Thus the first squared norm measures the mismatch between the input relative
modular evolution and the output relative modular evolution pulled back by
$\Phi^\dagger$, weighted by $\rho^{1/2}$.
For every $s\in\R$, the numerator of the second boundary term controls
the error of the rotated Petz recovery:
\begin{equation}\label{eq:comparison-trace-recovery}
 \|\rho-\mathcal R_s(\Phi(\rho))\|_1^2
 \le4\left(\HS{\rho^{1/2}-\Gamma(1/2+is)}^2
             +\varepsilon_1(s)\right).
\end{equation}
Consequently, the second boundary integral bounds the averaged squared recovery error.

Furthermore, we can use Green's theorem to turn the boundary integral to a surfacer integral:
\begin{equation}\label{sv:U-surface}
\begin{aligned}
 \Delta
 &=2\int_0^{1/2}\!\int_\R
 \frac{\sin(2\pi x)}{\cosh(2\pi t)-\cos(2\pi x)}
 \HS{\Gamma'(x+it)}^2\dd t\dd x + \pi\int_\R\left[
 \frac{\varepsilon_0(t)}{\cosh(2\pi t)-1}
 +\frac{\varepsilon_1(t)}{\cosh(2\pi t)+1}\right]\dd t.
\end{aligned}
\end{equation}
\end{theorem} 
The measured-relative-entropy recovery bound follows similarly below.
\begin{theorem}[Exact remainder of data processing for measured relative-entropy]\label{thm:measured}
For $\omega>0$, define the nonnegative boundary functional
\begin{equation}\label{eq:P-functional}
 \Pcal(\omega)=\pi\int_\R\left[
 \frac{\HS{\rho^{1/2}-\omega^{-is}\Gamma(-is)}^2+\varepsilon_0(s)}{\cosh(2\pi s)-1}
 +\frac{\HS{\rho^{1/2}-\omega^{1/2+is}\Gamma(1/2+is)}^2+\varepsilon_{1,\omega}(s)}{\cosh(2\pi s)+1}
 \right]\dd s.
\end{equation}
Then, for every $\omega>0$,
\begin{equation}\label{eq:master}
 D(\rho\|\sigma)-D(\Phi(\rho)\|\Phi(\sigma))=\Tr(\rho\log\omega)+1-\Tr(\tau\omega)+\Pcal(\omega).
\end{equation}
From the variational characterization of measured relative entropy~\cite{BFT,SutterBertaTomamichel2017},
\begin{equation}\label{eq:measured-variational-main}
 \DM(\rho\|\tau)
 =\sup_{\omega>0}\{\Tr(\rho\log\omega)+1-\Tr(\tau\omega)\}.
\end{equation}
we have
\begin{equation}\label{eq:measured-equality}
 D(\rho\|\sigma)-D(\Phi(\rho)\|\Phi(\sigma))-\DM(\rho\|\tau)
   =\inf_{\omega>0}\Pcal(\omega)\geq0.
\end{equation}
\end{theorem}
The known inequality $\DM(\rho\|\tau)\leq D(\rho\|\sigma)-D(\Phi(\rho)\|\Phi(\sigma))$ follows from the nonnegativity of this gap
\cite{SutterBertaTomamichel2017}; our equality gives its exact value.

\subsection{Conditional mutual information}\label{sec:CMI}
For a faithful tripartite state $\rho=\rho_{XYZ}$, take
$\Phi=\Tr_Z$ and $\sigma=\rho_X\otimes\rho_{YZ}$. One can identify the conditional mutual
as the relative-entropy loss
\[
 I(X:Z\mid Y)_\rho
 =D(\rho\|\rho_X\otimes\rho_{YZ})
  -D(\rho_{XY}\|\rho_X\otimes\rho_Y).
\]
Its nonnegativity is strong subadditivity~\cite{lieb1973proof}. Vanishing
CMI characterizes quantum Markov states~\cite{HaydenJozsaPetzWinter2004},
and small CMI also implies approximate recoverability~\cite{Fawzi2014}. Note that approximate recoverability is distinct from proximity to an exact quantum
Markov state: the minimum relative entropy to such states can be much larger
than the CMI~\cite{IbinsonLindenWinter2008}.

Here $\Phi^\dagger(A)=A\otimes\one_Z$ is a $*$-homomorphism, so
$\varepsilon_0(s)=\varepsilon_{1,\omega}(s)=0$ for every $\omega>0$
by Appendix~\ref{app:boundary-nonnegativity}.
Thus Theorems~\ref{thm:loss} and~\ref{thm:measured}, with
$\Delta=I(X:Z\mid Y)_\rho$, give the following corollary.

\begin{corollary}[CMI boundary, surface, and measured-recovery equalities]\label{cor:CMI}\label{sv:CMI-surface-thm}
Let
\begin{equation}\label{eq:CMI-Gamma}
 \Gamma^{XYZ}(z)
 =\rho_{YZ}^z\rho_Y^{-z}\rho_{XY}^z\rho^{1/2-z},
\end{equation}
and
\begin{equation}\label{eq:CMI-P}
 \Pcal_{XYZ}(\omega)=\pi\int_\R\left[
 \frac{\HS{\rho^{1/2}-\omega^{-is}\Gamma^{XYZ}(-is)}^2}{\cosh(2\pi s)-1}
 +\frac{\HS{\rho^{1/2}-\omega^{1/2+is}\Gamma^{XYZ}(1/2+is)}^2}{\cosh(2\pi s)+1}
 \right]\dd s.
\end{equation}
\begin{equation}\label{eq:CMI-tau}
 \tau=\int_\R\frac{\pi}{\cosh(2\pi s)+1}
         \rho_{YZ}^{1/2+is}\rho_Y^{-1/2-is}\rho_{XY}
         \rho_Y^{-1/2+is}\rho_{YZ}^{1/2-is}\dd s.
\end{equation}
Then
\begin{equation}\label{eq:CMI-main}
 I(X:Z\mid Y)_\rho=\Pcal_{XYZ}(\one),\qquad
 I(X:Z\mid Y)_\rho-\DM(\rho\|\tau)
   =\inf_{\omega>0}\Pcal_{XYZ}(\omega).
\end{equation}
The surface integral representation is
\begin{equation}\label{sv:CMI-surface}
 I(X:Z\mid Y)_\rho=2\int_0^{1/2}\!\int_\R
           \frac{\sin(2\pi x)}{\cosh(2\pi t)-\cos(2\pi x)}\|{\Gamma^{XYZ}}'(x+it)\|_2^2\dd t\dd x.
\end{equation}
By the product rule, the explicit derivative is
\begin{equation}\label{sv:CMI-derivative}
 {\Gamma^{XYZ}}'(z) =\rho_{YZ}^z(\log\rho_{YZ}-\log\rho_Y) \rho_Y^{-z}\rho_{XY}^z\rho^{1/2-z}+\rho_{YZ}^z\rho_Y^{-z}\rho_{XY}^z (\log\rho_{XY}-\log\rho)\rho^{1/2-z}.
\end{equation}
\end{corollary}

\subsection{Belavkin--Staszewski entropy-loss identities}\label{sec:bs-results}
\label{sec:bs-amplitudes}\label{sec:bs-integral}
We retain the faithful finite-dimensional setup above. The BS boundary
identity and fidelity remainder are proved in Appendix~\ref{sec:bs-poisson-bures};
the surface identity follows from the boundary identity in
Section~\ref{sv:surface-proof}. The Belavkin--Staszewski (BS) relative entropy is
\begin{equation}\label{eq:bs-definition}
 D_{\mathrm{BS}}(\rho\|\sigma)
 =\Tr\rho\log(\rho^{1/2}\sigma^{-1}\rho^{1/2})
 =\Tr\sigma\bigl(\sigma^{-1/2}\rho\sigma^{-1/2}\bigr)
       \log\bigl(\sigma^{-1/2}\rho\sigma^{-1/2}\bigr).
\end{equation}
We write
\begin{equation}\label{eq:bs-loss}
 \Delta_{\mathrm{BS}}
 :=D_{\mathrm{BS}}(\rho\|\sigma)
   -D_{\mathrm{BS}}\bigl(\Phi(\rho)\|\Phi(\sigma)\bigr).
\end{equation}
The identity below proves $\Delta_{\mathrm{BS}}\geq0$ directly and
retains its full two-boundary remainder.

Define the positive likelihood-ratio operators
\begin{equation}\label{eq:bs-likelihood-ratios}
 \ell_{\mathrm{BS}}:=\sigma^{-1/2}\rho\sigma^{-1/2},
 \qquad
 \widehat\ell_{\mathrm{BS}}
 :=\Phi(\sigma)^{-1/2}\Phi(\rho)\Phi(\sigma)^{-1/2}.
\end{equation}
Define the analytic input-space family
\begin{equation}\label{eq:bs-Gamma}
 \Gamma_{\mathrm{BS}}(z)
 :=\ell_{\mathrm{BS}}^{1/2-z}\sigma^{1/2}
 \Phi^\dagger\!\left(\Phi(\sigma)^{-1/2}
 \widehat\ell_{\mathrm{BS}}^{\,z}\Phi(\sigma)^{1/2}\right).
\end{equation}
Its value at zero is a factor of the original state:
\begin{equation}\label{eq:bs-reference-factorization}
 \Gamma_{\mathrm{BS}}(0)=\ell_{\mathrm{BS}}^{1/2}\sigma^{1/2},
 \qquad \Gamma_{\mathrm{BS}}(0)^\dagger\Gamma_{\mathrm{BS}}(0)=\rho.
\end{equation}
The boundary correction terms
\begin{equation}\label{eq:bs-defects}
 \delta_{\mathrm{BS},j}(s)
 :=1-\HS{\Gamma_{\mathrm{BS}}(j/2+is)}^2,
 \qquad j\in\{0,1\},
\end{equation}
are nonnegative by Lemma~\ref{lem:bs-contraction}.
\begin{theorem}[BS Poisson and surface identities]\label{thm:bs-poisson}\label{sv:BS-surface-thm}
Under the stated faithful finite-dimensional setup,
\begin{equation}\label{eq:bs-poisson}
 \Delta_{\mathrm{BS}}
 =\pi\int_{\R}\Bigg[
 \frac{\HS{\Gamma_{\mathrm{BS}}(0)-\Gamma_{\mathrm{BS}}(is)}^2
                   +\delta_{\mathrm{BS},0}(s)}
             {\cosh(2\pi s)-1}
 +\frac{\HS{\Gamma_{\mathrm{BS}}(0)-\Gamma_{\mathrm{BS}}(1/2+is)}^2
                   +\delta_{\mathrm{BS},1}(s)}
             {\cosh(2\pi s)+1}
 \Bigg]\dd s.
\end{equation}
\begin{equation}\label{sv:BS-surface}
\begin{aligned}
 \Delta_{\mathrm{BS}}
 &=2\int_0^{1/2}\!\int_\R
 \frac{\sin(2\pi x)}{\cosh(2\pi t)-\cos(2\pi x)}
 \HS{\Gamma_{\mathrm{BS}}'(x+it)}^2\dd t\dd x+\pi\int_\R\left[
 \frac{\delta_{\mathrm{BS},0}(t)}{\cosh(2\pi t)-1}
 +\frac{\delta_{\mathrm{BS},1}(t)}{\cosh(2\pi t)+1}\right]\dd t.
\end{aligned}
\end{equation}
Every displayed summand is nonnegative. All integrals converge
absolutely. No principal-value prescription is required.
\end{theorem}

\subsection{Fidelity remainder for the BS loss}\label{sec:bs-fidelity-variational}
\label{sec:bs-reconstruction}
The upper-boundary matrix defines a positive reconstructed operator
\begin{equation}\label{eq:bs-reconstructed-operator}
 \tau_{\mathrm{BS}}(s)
 :=\Gamma_{\mathrm{BS}}(1/2+is)^\dagger\Gamma_{\mathrm{BS}}(1/2+is)
 =\Phi^\dagger\!\left(
       \Phi(\sigma)^{1/2}\widehat\ell_{\mathrm{BS}}^{1/2-is}
       \Phi(\sigma)^{-1/2}\right)\sigma\,
   \Phi^\dagger\!\left(
       \Phi(\sigma)^{-1/2}\widehat\ell_{\mathrm{BS}}^{1/2+is}
       \Phi(\sigma)^{1/2}\right).
\end{equation}
In particular,
\begin{equation}\label{eq:bs-subnormalization}
 \tau_{\mathrm{BS}}(s)\geq0,
 \qquad
 \Tr\tau_{\mathrm{BS}}(s)=1-\delta_{\mathrm{BS},1}(s)\leq1.
\end{equation}
The averaged reconstruction
\begin{equation}\label{eq:bs-averaged-reconstruction}
 \overline\tau_{\mathrm{BS}}
 :=\int_{\R}\frac{\pi}{\cosh(2\pi s)+1}\tau_{\mathrm{BS}}(s)\dd s
\end{equation}
is also positive and subnormalized.

The construction depends only on $\sigma$, $\Phi$, and the accessible
output $\Phi(\rho)$, and returns $\sigma$ when $\Phi(\rho)=\Phi(\sigma)$.
It is generally nonlinear in $\Phi(\rho)$, as is the positive symmetrized
reconstruction of Ref.~\cite[Eq.~(94)]{BluhmCapelCostaRicoJencova2025},
\[
 \left[\sigma\Phi^\dagger\!\left(
 \Phi(\sigma)^{-1}\Phi(\rho)^2\Phi(\sigma)^{-1}
 \right)\sigma\right]^{1/2}.
\]
The two constructions differ in the placement of the operator powers and
the adjoint channel $\Phi^\dagger$, although both reconstruct $\rho$ when
$\Delta_{\mathrm{BS}}=0$.
In the Umegaki case, adding the positive correction in
\eqref{eq:recovery-schwarz-completion} to $\Gamma\Gamma^\dagger$
gives the rotated Petz recovered state.

\begin{theorem}[Variational fidelity remainder for the BS loss]
\label{thm:bs-fidelity-remainder}
For every $\omega>0$ define
\begin{equation}\label{eq:bs-fid-P-functional}
 \begin{aligned}
 \mathcal P_{\mathrm{BS}}(\omega)
 &:=
 \pi\int_{\R}
 \frac{
   \left\|\Gamma_{\mathrm{BS}}(0)-\bigl[\Tr(\rho\omega)\Tr(\overline\tau_{\mathrm{BS}}\omega^{-1})\bigr]^{-is}\Gamma_{\mathrm{BS}}(is)\right\|_2^2
   +\delta_{\mathrm{BS},0}(s)
 }{
   \cosh(2\pi s)-1
 }\dd s
 \\[2pt] &\quad+
 \pi\int_{\R}
 \frac{
 \left\|
   \frac{\Gamma_{\mathrm{BS}}(0)\omega^{1/2}}{\sqrt{\Tr(\rho\omega)}}
   -
   \bigl[\Tr(\rho\omega)\Tr(\overline\tau_{\mathrm{BS}}\omega^{-1})\bigr]^{-is}
   \frac{\Gamma_{\mathrm{BS}}(1/2+is)\omega^{-1/2}}{\sqrt{\Tr(\overline\tau_{\mathrm{BS}}\omega^{-1})}}
 \right\|_2^2
 }{
   \cosh(2\pi s)+1
 }\dd s .
 \end{aligned}
\end{equation}
Then $\mathcal P_{\mathrm{BS}}(\omega)\ge0$ and
\begin{equation}
 \Delta_{\mathrm{BS}}
 =
 -\log\!\left[
   \Tr(\rho\omega)\,
   \Tr(\overline\tau_{\mathrm{BS}}\omega^{-1})
 \right]
 +\mathcal P_{\mathrm{BS}}(\omega)
 \label{eq:bs-fid-test-identity}
\end{equation}
for every $\omega>0$. Using Alberti's variational formula for fidelity~\cite{Alberti1983},
\[
 F(\rho,\overline\tau_{\mathrm{BS}})
 =\inf_{\omega>0}\Tr(\rho\omega)\,
                   \Tr(\overline\tau_{\mathrm{BS}}\omega^{-1}),
\]
where $F(a,b):=\|a^{1/2}b^{1/2}\|_1^2$. We obtain
\begin{equation}
 \Delta_{\mathrm{BS}}
 =
 -\log F(\rho,\overline\tau_{\mathrm{BS}})
 +
 \inf_{\omega>0}\mathcal P_{\mathrm{BS}}(\omega)
 \label{eq:bs-fid-remainder}
\end{equation}

\end{theorem}

\section{Proof of main results}\label{sec:channel}
This section proves the Umegaki entropy-loss and CMI identities and derives
the BS surface identity from its boundary counterpart. The remaining BS
proofs are in Appendix~\ref{sec:bs-poisson-bures}.
\subsection{The strip Poisson formula}\label{sec:poisson}
We use the Poisson integration formula for a strip
\cite{widder1961functions}. With strip width $1/2$, it says
that a bounded harmonic function $h$, continuous on the closed strip,
satisfies, for $0<x<1/2$,
\begin{align}
 h(x)
 &=\int_\R
 \frac{\sin(2\pi x)}{\cosh(2\pi s)-\cos(2\pi x)}h(is)\dd s + \int_\R
 \frac{\sin(2\pi x)}{\cosh(2\pi s)+\cos(2\pi x)}h(1/2+is)\dd s.
 \label{eq:Poisson}
\end{align}
Its differentiated boundary form is the following:
\begin{lemma}[Boundary derivative]\label{lem:Poisson-derivative}
Let $h$ be as above, differentiable at $0$ in the inward normal direction,
with $h(0)=0$ and $h(is)=O(s^2)$ as $s\to0$. Then
\begin{equation}\label{eq:normal-derivative}
 \partial_x h(0)
 =2\int_\R \frac{\pi}{\cosh(2\pi s)-1}h(is)\dd s
  +2\int_\R\frac{\pi}{\cosh(2\pi s)+1}h(1/2+is)\dd s.
\end{equation}
\end{lemma}
\begin{proof}
Divide \eqref{eq:Poisson} by $x$ and let $x\downarrow0$. Away from $s=0$,
the two kernels divided by $x$ converge to
$2\pi/[\cosh(2\pi s)-1]$ and $2\pi/[\cosh(2\pi s)+1]$.
For sufficiently small $x>0$, the second kernel in \eqref{eq:Poisson}, divided by $x$,
is bounded by $Ce^{-2\pi|s|}$, uniformly in $x$, while near zero the first kernel satisfies
\[
 \frac{\sin(2\pi x)}
 {x\,[\cosh(2\pi s)-\cos(2\pi x)]}
 \leq\frac{C}{x^2+s^2}.
\]
Multiplication by $h(is)=O(s^2)$ gives a uniform local integrable bound.
The tails of the first integral are exponentially dominated because
$h$ is bounded. Dominated convergence proves \eqref{eq:normal-derivative}.
\end{proof}
\begin{remark}[Parameter convention]
The conventional recovery density is
$\beta_0(t)=\pi/[2(\cosh(\pi t)+1)]$ \cite{Junge2015,SutterBertaTomamichel2017}.
With $t=2s$, the weight used here is
$2\beta_0(2s)=\pi/[\cosh(2\pi s)+1]$; the corresponding rotation is
$\mathcal R_s$ in \eqref{eq:recovery-map}. 
\end{remark}
\subsection{The scalar analytic function}
Define
\begin{equation}\label{eq:scalar-f}
 f_\omega(z) =\Tr[\rho^{1/2}\omega^{z}\Gamma(z)] =\Tr\!\left[\rho^{1-z}\omega^z\sigma^z \Phi^\dagger(\Phi(\sigma)^{-z}\Phi(\rho)^z)\right].
\end{equation}
The second equality is only a cyclic permutation in the trace; it does
not commute any distinct matrix factors. We have $f_\omega(0)=1$ and
\begin{align}
 f_\omega'(0)
 &=\Tr\rho(-\log\rho+\log\omega+\log\sigma)
    +\Tr\rho\Phi^\dagger(-\log\Phi(\sigma)+\log\Phi(\rho))\notag\\
 &=\Tr(\rho\log\omega)-\bigl[D(\rho\|\sigma)-D(\Phi(\rho)\|\Phi(\sigma))\bigr]\in\R.
 \label{eq:derivative-f}
\end{align}
Set
\begin{equation}\label{eq:h}
 h_\omega(z)=1-\Re f_\omega(z).
\end{equation}
Since $f_\omega$ is analytic, its real part is harmonic by the
Cauchy--Riemann equations. Hence $h_\omega=1-\Re f_\omega$ is harmonic
and is bounded because $f_\omega$ is bounded. Moreover,
\begin{equation}\label{eq:h-derivative}
 h_\omega(0)=0,\qquad
 \partial_xh_\omega(0)=D(\rho\|\sigma)-D(\Phi(\rho)\|\Phi(\sigma))-\Tr(\rho\log\omega),\qquad
 h_\omega(is)=O(s^2).
\end{equation}
The last estimate follows from analyticity and the reality of
$f_\omega'(0)$: the linear term $is\,f_\omega'(0)$ has zero real part.
For any $z$ on the boundary of the strip, expanding the squared
Hilbert--Schmidt norm gives
\begin{equation}\label{eq:polarization}
 2h_\omega(z)=\HS{\rho^{1/2}-\omega^{z}\Gamma(z)}^2+1-\HS{\omega^{z}\Gamma(z)}^2.
\end{equation}
By \eqref{eq:defect-norms} and unitary invariance of the Hilbert--Schmidt norm,
\begin{align}
 2h_\omega(-is)
 &=\HS{\rho^{1/2}-\omega^{-is}\Gamma(-is)}^2+\varepsilon_0(s),
      \label{eq:lower-polarization}\\
 2h_\omega(1/2+is)
 &=\HS{\rho^{1/2}-\omega^{1/2+is}\Gamma(1/2+is)}^2
    +\varepsilon_{1,\omega}(s)+1-\Tr(\omega\mathcal R_s(\Phi(\rho))).
      \label{eq:upper-polarization}
\end{align}
By \eqref{eq:h-derivative}, $h_\omega(-is)=O(s^2)$.
Since both terms on the right of \eqref{eq:lower-polarization} are
nonnegative, each is $O(s^2)$ separately. Together with
$\cosh(2\pi s)-1\sim2\pi^2s^2$, this makes the lower-boundary
integrand bounded near zero. The upper-boundary denominator is at least
two. Both numerators are bounded on the real axis and both kernels decay
exponentially at infinity, so the integrals in \eqref{eq:P-functional}
converge absolutely.
\begin{proof}[Proof of Theorem~\ref{thm:loss}]
Apply Lemma~\ref{lem:Poisson-derivative} to $h_\omega$. Since $\pi/[\cosh(2\pi s)-1]$ is even,
the integral along $\operatorname{Re}z=0$ may be parametrized by $-is$. Substituting
\eqref{eq:lower-polarization}--\eqref{eq:upper-polarization} gives
\begin{equation*}
 D(\rho\|\sigma)-D(\Phi(\rho)\|\Phi(\sigma))-\Tr(\rho\log\omega) =\Pcal(\omega) +\int_\R\frac{\pi}{\cosh(2\pi s)+1}\,[1-\Tr(\omega\mathcal R_s(\Phi(\rho)))]\dd s =\Pcal(\omega)+1-\Tr(\tau\omega).
\end{equation*}
This proves \eqref{eq:master}. With $\omega=\one_A$, the scalar terms
cancel because $\Tr\tau=1$, and \eqref{eq:loss-expanded} follows.
The preceding estimates justify convergence and the removable
singularity.

To prove \eqref{eq:comparison-trace-recovery}, set
$D_s=\mathcal R_s(\Phi(\rho))-\Gamma(1/2+is)\Gamma(1/2+is)^\dagger\ge0$;
Appendix~\ref{app:boundary-nonnegativity} gives $\Tr D_s=\varepsilon_1(s)$.
For $A=[\rho^{1/2}\;0]$ and
$B=[\Gamma(1/2+is)\;D_s^{1/2}]$, we have
$AA^\dagger=\rho$, $BB^\dagger=\mathcal R_s(\Phi(\rho))$.
Thus H\"older's inequality yields
\[
 \|\rho-\mathcal R_s(\Phi(\rho))\|_1
 =\|(A-B)A^\dagger+B(A-B)^\dagger\|_1
 \le (\HS A+\HS B)\HS{A-B}=2\HS{A-B}.
\]
Squaring and using $\HS{A-B}^2=\HS{\rho^{1/2}-\Gamma(1/2+is)}^2+\varepsilon_1(s)$
proves \eqref{eq:comparison-trace-recovery}.
\end{proof}
\subsection{The gap in the measured-relative-entropy recovery bound}\label{sec:measured}
We recall the measured-relative-entropy variational formula
from Ref.~\cite{BFT} (Lemma~1 and Theorem~2); see also
Ref.~\cite[Eq.~(49)]{SutterBertaTomamichel2017}.
\begin{lemma}[Measured variational formula]\label{lem:measured}
For faithful density matrices $a,b$,
\begin{equation}\label{eq:variational}
 \DM(a\|b)=\sup_{\omega>0}
   \{\Tr(a\log\omega)+1-\Tr(b\omega)\}.
\end{equation}
A maximizer exists and every maximizer satisfies $\Tr(b\omega)=1$.
\end{lemma}

\begin{proof}[Proof of Theorem~\ref{thm:measured}]
The test-operator identity was proved in Section~\ref{sec:channel}.
Apply Lemma~\ref{lem:measured} to the faithful states $a=\rho$ and $b=\tau$.
Equation~\eqref{eq:master} implies, for every
$\omega>0$,
\[
 \Pcal(\omega)=D(\rho\|\sigma)-D(\Phi(\rho)\|\Phi(\sigma))-
 \{\Tr(\rho\log\omega)+1-\Tr(\tau\omega)\}.
\]
Taking the infimum on the left and the supremum on the right gives
\eqref{eq:measured-equality} and nonnegativity.
\end{proof}

\subsection{A vector-valued Green identity and the surface integral representation}\label{sv:surface-proof}
For a bounded analytic function $F$ with values in a finite-dimensional
Hilbert space, analytic near the finite points of the closed strip, write
$z=x+it$ and set
\begin{equation}\label{sv:energy}
 \mathscr E(F)=\int_\R \frac{\pi}{\cosh(2\pi t)-1}\|F(0)-F(it)\|^2\dd t
       +\int_\R \frac{\pi}{\cosh(2\pi t)+1}\|F(0)-F(1/2+it)\|^2\dd t.
\end{equation}
\begin{lemma}[Boundary integral equals surface integral]\label{sv:greenlemma}
Under these assumptions,
\begin{equation}\label{sv:greenidentity}
 \mathscr E(F)=2\int_0^{1/2}\!\int_\R
           \frac{\sin(2\pi x)}{\cosh(2\pi t)-\cos(2\pi x)}\|F'(x+it)\|^2\dd t\dd x.
\end{equation}
\end{lemma}
\begin{proof}
Set $q(z)=\|F(z)-F(0)\|^2$. Analyticity gives
$\partial_xF=F'$ and $\partial_tF=iF'$, while each component
of $F$ is harmonic. Hence
\[
 (\partial_x^2+\partial_t^2)q
 =2\|\partial_xF\|^2+2\|\partial_tF\|^2
 =4\|F'\|^2.
\]
We compare $q$ with its bounded harmonic extension $H_q$ from
 the boundary of the strip. By the strip Poisson formula~\cite{widder1961functions},
\[
\begin{aligned}
 H_q(r)
 &=\int_\R
 \frac{\sin(2\pi r)}
      {\cosh(2\pi t)-\cos(2\pi r)}q(it)\dd t +\int_\R
 \frac{\sin(2\pi r)}
      {\cosh(2\pi t)+\cos(2\pi r)}q(1/2+it)\dd t,
 \qquad 0<r<1/2.
\end{aligned}
\]
Thus $H_q-q$ vanishes on both boundaries and satisfies
\[
 -(\partial_x^2+\partial_t^2)(H_q-q)=4\|F'\|^2.
\]
The Dirichlet Green function therefore expresses this difference
as an integral over the strip:
\[
 H_q(r)-q(r)
 =4\int_0^{1/2}\!\int_\R
 G(r,x+it)\|F'(x+it)\|^2\dd t\dd x,
\]
where
\begin{equation}\label{sv:Green}
 G(r,x+it)=\frac1{4\pi}\log
 \frac{\cosh(2\pi t)-\cos(2\pi(r+x))}
      {\cosh(2\pi t)-\cos(2\pi(r-x))}.
\end{equation}
This is the strip Green function in Ref.~\cite{Melnikov2012},
rescaled to width $1/2$, with convention
$-(\partial_x^2+\partial_t^2)G=\delta(x-r)\delta(t)$.
The representation follows from Green's second identity:
the boundary terms on $x=0,1/2$ vanish, and those at
$t=\pm T$ tend to zero as $T\to\infty$ by the exponential
decay of $G$ and its derivatives and the boundedness of
$H_q-q$ and its first derivatives.

It remains to divide this representation by $r$ and let
$r\downarrow0$. Analyticity at the origin gives
$q(r)=O(r^2)$ and $q(it)=O(t^2)$. Dividing the Poisson
formula above by $r$ and applying the dominated-convergence
argument of Lemma~\ref{lem:Poisson-derivative} yields
\[
 \lim_{r\downarrow0}\frac{H_q(r)-q(r)}r
 =2\mathscr E(F), \qquad \frac{G(r,x+it)}r
 \longrightarrow
 \frac{\sin(2\pi x)}
      {\cosh(2\pi t)-\cos(2\pi x)}.
\]
By Appendix~\ref{app:green-boundary-limit}, this convergence
holds in $L^1$ over the strip. Since $\|F'\|^2$ is bounded,
we may pass to the limit in the Green representation to obtain
\[
 2\mathscr E(F)
 =4\int_0^{1/2}\!\int_\R
 \frac{\sin(2\pi x)}
      {\cosh(2\pi t)-\cos(2\pi x)}
 \|F'(x+it)\|^2\dd t\dd x.
\]
Dividing by two proves \eqref{sv:greenidentity}.
\end{proof}

\begin{proof}[Proof of the surface identity in Theorem~\ref{thm:loss}]
Apply Lemma~\ref{sv:greenlemma} to $\Gamma$, using
$\Gamma(0)=\rho^{1/2}$. Together with the boundary identity
\eqref{eq:loss-expanded}, this gives
\begin{equation}\label{sv:Gamma-area-defects}
\begin{aligned}
 \Delta
 &=2\int_0^{1/2}\!\int_\R
 \frac{\sin(2\pi x)}{\cosh(2\pi t)-\cos(2\pi x)}
 \HS{\Gamma'(x+it)}^2\dd t\dd x+\pi\int_\R\left[
 \frac{\varepsilon_0(t)}{\cosh(2\pi t)-1}
 +\frac{\varepsilon_1(t)}{\cosh(2\pi t)+1}\right]\dd t.
\end{aligned}
\end{equation}
This is \eqref{sv:U-surface}. Nonnegativity follows from the squared
norm and the boundary Schwarz inequalities.
\end{proof}

\begin{proof}[Proof of the BS surface identity \eqref{sv:BS-surface}]
The same argument applies to the bounded analytic family
$\Gamma_{\mathrm{BS}}$ in \eqref{eq:bs-Gamma}, with
$\Gamma_{\mathrm{BS}}(0)=\ell_{\mathrm{BS}}^{1/2}\sigma^{1/2}$.
The boundary loss identity \eqref{eq:bs-poisson} reads
\[
 \Delta_{\mathrm{BS}}=\mathscr E(\Gamma_{\mathrm{BS}})
 +\pi\int_\R\left[
 \frac{\delta_{\mathrm{BS},0}(t)}{\cosh(2\pi t)-1}
 +\frac{\delta_{\mathrm{BS},1}(t)}{\cosh(2\pi t)+1}\right]\dd t.
\]
Lemma~\ref{sv:greenlemma} converts $\mathscr E(\Gamma_{\mathrm{BS}})$
into the weighted surface integral of $\HS{\Gamma_{\mathrm{BS}}'}^2$,
giving \eqref{sv:BS-surface}. The two nonnegative correction terms
remain boundary integrals, exactly as in \eqref{sv:Gamma-area-defects}.
\end{proof}

\section{Pointwise estimate}\label{sv:pointwise}
We bound the recovery errors at any prescribed rotation in this section and 
compare the bounds with previous quantitative estimates.

\subsection{Prescribed-parameter bounds with spectral-count and logarithmic dependence}\label{sec:frequency-lemma-proof}
We first establish frequency-evaluation inequalities
that convert the integrated boundary errors into pointwise bounds.
The following estimate follows from Cauchy--Schwarz and the classical
entry-sum identity for inverse Cauchy matrices; see
\cite[Theorem~2.1]{GrinbergInverseCauchy} and  \cite{Schechter1959}. The proofs of the following four lemmas are given in
Appendix~\ref{app:frequency-evaluation-proofs}.
\begin{lemma}[One-sided finite-frequency evaluation]\label{lem:one-sided-evaluation}
Let $\omega_1,\ldots,\omega_N$ be distinct real numbers and let
$p(t)=\sum_{j=1}^N c_je^{i\omega_jt}$, with $c_j\in\C$.
For every $a>0$,
\begin{equation}\label{eq:one-sided-evaluation}
 |p(0)|^2\le aN\int_0^\infty e^{-at}|p(t)|^2\dd t.
\end{equation}
The constant $aN$ is optimal for each such frequency set.
\end{lemma}

\begin{lemma}[Two finite-frequency evaluation inequalities]\label{lem:frequency-evaluation}
Let $p(s)=\sum_{j=1}^N c_j e^{i\omega_js}$, where the $\omega_j$
are distinct real numbers and $c_j\in\C$. For every $s_0\in\R$,
\begin{equation}\label{eq:finite-frequency-evaluation}
 |p(s_0)|^2\le 2N\cosh^2(\pi s_0)
 \int_\R\frac{\pi|p(s)|^2}{\cosh(2\pi s)+1}\dd s.
\end{equation}
If additionally $p(0)=0$, then
\begin{equation}\label{eq:finite-frequency-vanishing}
 |p(s_0)|^2\le N\sinh(2\pi|s_0|)
 \int_\R\frac{\pi|p(s)|^2}{\cosh(2\pi s)-1}\dd s.
\end{equation}
\end{lemma}

\begin{lemma}[Two-sided exponential-weight evaluation]
\label{lem:pwlog-exponential}
Let
\begin{equation}
 p(t)=c_0+\sum_{j=1}^{N}c_j e^{i\omega_jt},
 \qquad N\geq1,
 \label{eq:pwlog-polynomial}
\end{equation}
where the $\omega_j$ are real and lie in an interval of length at most
$W\geq0$. For $a>0$, define
\begin{equation}
 A_a(W,N)=\left[
 \sqrt{\frac{8a}{\pi^2}}+
 \sqrt{\frac{4W}{\pi}
       +\frac{4a}{\pi^2}\bigl(2\log(2N)+1\bigr)}
 \right]^2.
 \label{eq:pwlog-A}
\end{equation}
Then, for every $s_0\in\R$,
\begin{equation}
 |p(s_0)|^2\leq A_a(W,N)
 \int_{\R} e^{-a|t-s_0|}|p(t)|^2\dd t.
 \label{eq:pwlog-exponential}
\end{equation}
\end{lemma}

For a faithful
density matrix $a$, write $\Lambda_a=\operatorname{spec}(\log a)$ and
$r_a=|\Lambda_a|$, counting distinct eigenvalues, and let
\begin{equation}
 w(a)=\log\frac{\lambda_{\max}(a)}{\lambda_{\min}(a)}.
 \label{eq:pwlog-spectral-data}
\end{equation}

\begin{lemma}[Evaluation against the two Poisson weights]
\label{lem:pwlog-Poisson}
Under the hypotheses of Lemma~\ref{lem:pwlog-exponential}, define
\begin{equation}
 C_{\mathrm{eval}}(W,N)=\frac{8}{\pi^2}
 \left[2+\sqrt{W+4\log(2N)+2}\right]^2.
 \label{eq:pwlog-K}
\end{equation}
For every $s_0\in\R$,
\begin{equation}
 |p(s_0)|^2\leq C_{\mathrm{eval}}(W,N)\cosh^2(\pi s_0)
 \int_{\R}\frac{\pi|p(t)|^2}{\cosh(2\pi t)+1}\dd t.
 \label{eq:pwlog-plus-evaluation}
\end{equation}
If in addition $p(0)=0$, then
\begin{align}
 |p(s_0)|^2
 \leq\frac12 C_{\mathrm{eval}}(W,N)\sinh(2\pi|s_0|)
 \int_{\R}\frac{\pi|p(t)|^2}{\cosh(2\pi t)-1}\dd t.
 \label{eq:pwlog-minus-evaluation}
\end{align}
For the faithful density matrices $\sigma$ and $\Phi(\sigma)$ in our setup,
\begin{equation}
 C_{\mathrm{eval}}\!\left(w(\sigma)+w(\Phi(\sigma)),r_\sigma r_{\Phi(\sigma)}\right)
 \leq16(1+\op{\log\sigma}+\op{\log\Phi(\sigma)}).
 \label{eq:pwlog-log-simplification}
\end{equation}
\end{lemma}

We now apply these inequalities to get the recovery error bounds at prescribed rotations.
\begin{theorem}[Pointwise bounds at prescribed real parameters]
\label{thm:finite-recovery}
Write $C_{\mathrm{eval}}$ for the reference-state coefficient on the
left-hand side of \eqref{eq:pwlog-log-simplification}.
For $s\in\R$, define $C_+(s)$ below, and for $s\ne0$ define $C_-(s)$:
\begin{equation}\label{eq:pwlog-upper-coefficient}
\begin{aligned}
 C_+(s)&=\min\left\{2(r_\rho+r_{\Phi(\sigma)}),
 C_{\mathrm{eval}}\right\}
 \cosh^2(\pi s).
\end{aligned}
\end{equation}
\begin{equation}\label{eq:pwlog-lower-coefficient}
\begin{aligned}
 C_-(s)&=\min\left\{r_\sigma+r_{\Phi(\rho)},
 \tfrac12 C_{\mathrm{eval}}\right\}\sinh(2\pi|s|).
\end{aligned}
\end{equation}
For every prescribed $t\in\R$ and $s\ne0$,
\begin{equation}\label{eq:finite-recovery-simple}
 \begin{aligned}
 \HS{\rho^{1/2}-\Gamma(1/2+it)}^2+\varepsilon_1(t)&\le C_+(t)\int_\R\frac{\pi\bigl(\HS{\rho^{1/2}-\Gamma(1/2+iu)}^2+\varepsilon_1(u)\bigr)}{\cosh(2\pi u)+1}\dd u\le C_+(t)\Delta,
 \end{aligned}
\end{equation}
\begin{equation}\label{eq:compressed-lower-simple}
 \begin{aligned}
 \HS{\rho^{1/2}-\Gamma(-is)}^2+\varepsilon_0(s)&\le C_-(s)\int_\R\frac{\pi\bigl(\HS{\rho^{1/2}-\Gamma(-iu)}^2+\varepsilon_0(u)\bigr)}{\cosh(2\pi u)-1}\dd u \le C_-(s)\Delta.
 \end{aligned}
\end{equation}
The coefficient $C_{\mathrm{eval}}$ depends only on the reference states
$\sigma$ and $\Phi(\sigma)$ and obeys the logarithmic bound
\eqref{eq:pwlog-log-simplification}.
\end{theorem}
\label{sec:finite-recovery-proof}
\label{sec:finite-modular-proof}
\begin{proof}

Write
$\Phi(X)=\sum_jK_jXK_j^*$ within this proof. Define
\begin{equation}
 D_j^+(s)=\Phi(\rho)^{is}K_j\rho^{1/2-is}
 -\Phi(\rho)^{1/2}\Phi(\sigma)^{-1/2+is}
  K_j\sigma^{1/2-is}.
 \label{eq:pwlog-upper-Kraus}
\end{equation}
Both families in this difference have total squared norm one. For
the first this follows from $\sum_jK_j^*K_j=\one$. For the second,
\begin{equation*}
 \sum_j\HS{\Phi(\rho)^{1/2}\Phi(\sigma)^{-1/2+is} K_j\sigma^{1/2-is}}^2 =\Tr\!\left[ \Phi(\rho)^{1/2}\Phi(\sigma)^{-1/2+is} \left(\sum_jK_j\sigma K_j^*\right) \Phi(\sigma)^{-1/2-is}\Phi(\rho)^{1/2}\right]=1.
\end{equation*}
Expanding the squared norms and using
$\varepsilon_1(s)=1-\HS{\Gamma(1/2+is)}^2$ therefore gives
\begin{equation}
 \HS{\rho^{1/2}-\Gamma(1/2+is)}^2+\varepsilon_1(s)
 =\sum_j\HS{D_j^+(s)}^2.
 \label{eq:pwlog-upper-lifting}
\end{equation}
For the lower boundary, set
\begin{equation}
 D_j^-(s)=\Phi(\sigma)^{-is}K_j\sigma^{is}\rho^{1/2}
             -\Phi(\rho)^{-is}K_j\rho^{1/2+is}.
 \label{eq:pwlog-lower-Kraus}
\end{equation}
Unitarity and $\sum_jK_j^*K_j=\one$ show that both
families have total squared norm one. Expanding the cross term gives
\begin{equation}
 \HS{\rho^{1/2}-\Gamma(-is)}^2+\varepsilon_0(s)
 =\sum_j\HS{D_j^-(s)}^2.
 \label{eq:pwlog-lower-lifting}
\end{equation}
Here the summed cross term is
\[
 \Tr\!\left[\rho^{1/2}\sigma^{-is}
 \Phi^\dagger(\Phi(\sigma)^{is}\Phi(\rho)^{-is})
 \rho^{1/2+is}\right]
 =\Tr\rho^{1/2}\Gamma(-is).
\]

For the spectral-count estimate, project $D_j^+(s)$ on the left onto an
 eigenspace of $\log\Phi(\rho)$ with eigenvalue $d$ and on the right onto
 an eigenspace of $\log\sigma$ with eigenvalue $a$. The first term then
 has frequencies $d-b$, $b\in\Lambda_\rho$, and the second has frequencies
 $c-a$, $c\in\Lambda_{\Phi(\sigma)}$.
Writing $P_b^\rho,P_a^\sigma$ and
$Q_d^{\Phi(\rho)},Q_c^{\Phi(\sigma)}$ for the corresponding spectral
projections, we have
\begin{equation}\label{eq:pwlog-upper-spectral-expansion}
\begin{aligned}
 Q_d^{\Phi(\rho)}D_j^+(s)P_a^\sigma
 &=\sum_{b\in\Lambda_\rho}
 e^{b/2}e^{is(d-b)}
 Q_d^{\Phi(\rho)}K_jP_b^\rho P_a^\sigma
 -\sum_{c\in\Lambda_{\Phi(\sigma)}}
 e^{(d-c+a)/2}e^{is(c-a)}
 Q_d^{\Phi(\rho)}Q_c^{\Phi(\sigma)}K_jP_a^\sigma.
\end{aligned}
\end{equation}
Thus each block has at most
 $r_\rho+r_{\Phi(\sigma)}$ distinct frequencies. Apply
 \eqref{eq:finite-frequency-evaluation} entrywise and sum over blocks and
 Kraus operators. For any two complete orthogonal spectral resolutions,
 $\sum_{a,d}\|Q_dMP_a\|_2^2=\|M\|_2^2$, so this introduces no additional
 factor and gives the coefficient
 $2(r_\rho+r_{\Phi(\sigma)})\cosh^2(\pi s_0)$ multiplying the upper-boundary integral.

For the lower boundary, project $D_j^-(s)$ on the left onto an
 eigenspace of $\log\Phi(\sigma)$ with eigenvalue $c$ and on the right onto
 an eigenspace of $\log\rho$ with eigenvalue $b$. The resulting frequencies
 are $a-c$, $a\in\Lambda_\sigma$, and $b-d$, $d\in\Lambda_{\Phi(\rho)}$.
\begin{equation}\label{eq:pwlog-lower-spectral-expansion}
\begin{aligned}
 Q_c^{\Phi(\sigma)}D_j^-(s)P_b^\rho
 &=\sum_{a\in\Lambda_\sigma}
 e^{b/2}e^{is(a-c)}
 Q_c^{\Phi(\sigma)}K_jP_a^\sigma P_b^\rho -\sum_{d\in\Lambda_{\Phi(\rho)}}
 e^{b/2}e^{is(b-d)}
 Q_c^{\Phi(\sigma)}Q_d^{\Phi(\rho)}K_jP_b^\rho.
\end{aligned}
\end{equation}
 Each block has at most $r_\sigma+r_{\Phi(\rho)}$ distinct frequencies
 and vanishes at $s=0$, since $D_j^-(0)=0$. Applying
 \eqref{eq:finite-frequency-vanishing} entrywise and summing over the
 orthogonal blocks and Kraus operators gives the coefficient
 $(r_\sigma+r_{\Phi(\rho)})\sinh(2\pi|s_0|)$ multiplying the lower-boundary integral.

For the logarithmic estimate, multiply on the left by $\Phi(\rho)^{-is}$
and on the right by $\rho^{is}$ before projecting onto eigenspaces:
\begin{equation}\label{eq:pwlog-upper-frame}
 \widetilde D_j^+(s) =\Phi(\rho)^{-is}D_j^+(s)\rho^{is} =K_j\rho^{1/2} -\Phi(\rho)^{1/2-is}\Phi(\sigma)^{-1/2+is} K_j\sigma^{1/2-is}\rho^{is}.
\end{equation}
Its Hilbert--Schmidt norm equals that of $D_j^+(s)$.
Let $P_b^\rho$, $Q_d^{\Phi(\rho)}$, $P_a^\sigma$ and
$P_c^{\Phi(\sigma)}$ be the spectral projections of the respective
logarithms, with eigenvalues $b,d,a,c$. Projecting
\eqref{eq:pwlog-upper-frame} on the left by $Q_d^{\Phi(\rho)}$ and
on the right by $P_b^\rho$ gives
\begin{equation}\label{eq:pwlog-upper-block}
 Q_d^{\Phi(\rho)}\widetilde D_j^+(s)P_b^\rho =e^{b/2}Q_d^{\Phi(\rho)}K_jP_b^\rho-\sum_{a,c}e^{(d-c+a)/2} Q_d^{\Phi(\rho)}P_c^{\Phi(\sigma)}K_jP_a^\sigma P_b^\rho e^{i(b-d+c-a)s}.
\end{equation}
Every coefficient is independent of $s$, and no distinct matrix
factors have been commuted. The nonconstant frequencies belong to
\[
 (b-d)+\{c-a:
 a\in\operatorname{spec}(\log\sigma),\
 c\in\operatorname{spec}(\log\Phi(\sigma))\}.
\]
The frequency cluster is indexed by pairs of reference-state eigenvalues
and has width at most $w(\sigma)+w(\Phi(\sigma))$; $b,d$ only translate it.
Thus \eqref{eq:pwlog-plus-evaluation}, with monotonicity of the evaluation
coefficient in the width and frequency count, applies to each block
with the reference-state coefficient $C_{\mathrm{eval}}$.
Summing the block inequalities and then summing over $j$ recovers
\eqref{eq:pwlog-upper-lifting}, both pointwise and inside the
integral. Taking the minimum of this coefficient and the spectral-count
coefficient proves \eqref{eq:finite-recovery-simple}. 

For the lower boundary, multiply on the left by $\Phi(\rho)^{is}$
and on the right by $\rho^{-is}$:
\begin{equation}\label{eq:pwlog-lower-frame}
 \widetilde D_j^-(s) =\Phi(\rho)^{is}D_j^-(s)\rho^{-is} =\Phi(\rho)^{is}\Phi(\sigma)^{-is} K_j\sigma^{is}\rho^{1/2-is}-K_j\rho^{1/2}.
\end{equation}
With the same spectral projections,
\begin{equation}\label{eq:pwlog-lower-block}
 Q_d^{\Phi(\rho)}\widetilde D_j^-(s)P_b^\rho =\sum_{a,c}e^{b/2} Q_d^{\Phi(\rho)}P_c^{\Phi(\sigma)}K_jP_a^\sigma P_b^\rho e^{i(d-b+a-c)s}-e^{b/2}Q_d^{\Phi(\rho)}K_jP_b^\rho.
\end{equation}
This is a constant plus a frequency cluster with the same width and
frequency-count bounds as above.
Completeness of $P_a^\sigma$ and $P_c^{\Phi(\sigma)}$ gives
$Q_d^{\Phi(\rho)}\widetilde D_j^-(0)P_b^\rho=0.$
Applying \eqref{eq:pwlog-minus-evaluation} to each block gives the
coefficient $\tfrac12 C_{\mathrm{eval}}\sinh(2\pi|s_0|)$.
Summing over $b,d,j$ and using orthogonality of the spectral blocks,
unitary invariance, and \eqref{eq:pwlog-lower-lifting} yields the
logarithmic bound. Taking the minimum with the spectral-count bound
establishes the first inequality in \eqref{eq:compressed-lower-simple}.

\end{proof}

\begin{corollary}[Pointwise BS boundary bounds]\label{cor:bs-pointwise-boundary}
Under the faithful finite-dimensional assumptions of
Theorem~\ref{thm:bs-poisson}, let
$N=r_{\widehat\ell_{\mathrm{BS}}}$ count the distinct eigenvalues of
$\widehat\ell_{\mathrm{BS}}$ and set
\begin{equation}\label{eq:bs-pointwise-coefficient}
 W=\log\frac{\lambda_{\max}(\widehat\ell_{\mathrm{BS}})}
                  {\lambda_{\min}(\widehat\ell_{\mathrm{BS}})},
 \qquad
 C_{\mathrm{BS}}=\min\{2(N+1),C_{\mathrm{eval}}(W,N)\},
\end{equation}
where $C_{\mathrm{eval}}$ is defined in \eqref{eq:pwlog-K}.
For every $s\in\R$,
\begin{align}
 \HS{\Gamma_{\mathrm{BS}}(0)-\Gamma_{\mathrm{BS}}(1/2+is)}^2
 +\delta_{\mathrm{BS},1}(s)
 &\leq C_{\mathrm{BS}}\cosh^2(\pi s)\,\Delta_{\mathrm{BS}},
 \label{eq:bs-pointwise-upper}\\
 \HS{\Gamma_{\mathrm{BS}}(0)-\Gamma_{\mathrm{BS}}(is)}^2
 +\delta_{\mathrm{BS},0}(s)
 &\leq\tfrac12 C_{\mathrm{BS}}\sinh(2\pi|s|)\,\Delta_{\mathrm{BS}}.
 \label{eq:bs-pointwise-lower}
\end{align}
Each term on the left is nonnegative and hence obeys the same bound
separately. The coefficients depend only on the spectral count or
logarithmic spectral width of the output likelihood ratio.
\end{corollary}
\begin{proof}
Apply Lemmas~\ref{lem:frequency-evaluation} and~\ref{lem:pwlog-Poisson}
as in the proof of Theorem~\ref{thm:finite-recovery}.
Indeed, write $\Phi^\dagger(X)=V^\dagger(X\otimes I)V$ with $V$ an
isometry. With $Z_{\mathrm{BS}}$ from \eqref{eq:bs-output-family},
the Schwarz identities \eqref{eq:bs-schwarz-lower}--\eqref{eq:bs-schwarz-upper}
give, for $j=0,1$,
\[
 \delta_{\mathrm{BS},j}(s)
 =\left\|\ell_{\mathrm{BS}}^{(1-j)/2-is}\sigma^{1/2}
 V^\dagger\bigl(Z_{\mathrm{BS}}(j/2+is)\otimes I\bigr)
 (I-VV^\dagger)\right\|_2^2.
\]
After left projection onto an eigenspace of $\ell_{\mathrm{BS}}$
with eigenvalue $a$, every entry in this factor and in
$\Gamma_{\mathrm{BS}}(0)-\Gamma_{\mathrm{BS}}(j/2+is)$ is a constant
plus at most $N$ exponentials with frequencies $\log b-\log a$,
$b\in\operatorname{spec}(\widehat\ell_{\mathrm{BS}})$, lying in an
interval of width $W$. For $j=0$, all these entries vanish at $s=0$.
Apply the two evaluation lemmas entrywise and sum over the orthogonal
spectral projections. The corresponding boundary integral is at most
$\Delta_{\mathrm{BS}}$ by \eqref{eq:bs-poisson}, proving both bounds;
the lower bound at zero follows directly from
$\Gamma_{\mathrm{BS}}(0)$ and $\delta_{\mathrm{BS},0}(0)=0$.
\end{proof}

\subsection{Interior bound}\label{sec:interior-pointwise}
The two boundary numerators in \eqref{eq:loss-expanded} are boundary
values of the same harmonic function. The Poisson formula therefore
bounds the sum $\HS{\rho^{1/2}-\Gamma(z)}^2+1-\HS{\Gamma(z)}^2$
in the interior in terms of the entropy loss $\Delta$.

\begin{theorem}\label{sv:complex-interior}
For $z=x+it$ with $0<x<1/2$,
\begin{equation}\label{sv:complex-values}
 0\leq\HS{\rho^{1/2}-\Gamma(z)}^2+1-\HS{\Gamma(z)}^2
 \leq\frac{2\bigl[\cosh(2\pi t)-\cos(2\pi x)\bigr]}
 {\pi\sin(2\pi x)}\,\Delta.
\end{equation}
Both terms on the left are nonnegative. The coefficient depends only on the interior point $z$, not on the
dimensions or the spectra of the states.
In particular, at $z=1/4$,
\begin{equation}\label{sv:interior-midpoint}
 \HS{\rho^{1/2}-\Gamma(1/4)}^2\leq\frac{2\Delta}{\pi},
 \qquad
 \HS{\Gamma(1/4)}^2\geq1-\frac{2\Delta}{\pi}.
\end{equation}
\end{theorem}
\begin{proof}
First, $\HS{\Gamma(z)}^2$ is bounded and subharmonic, since
\[
 (\partial_x^2+\partial_t^2)\HS{\Gamma(x+it)}^2
 =4\HS{\Gamma'(x+it)}^2\geq0.
\]
The nonnegativity of the boundary correction terms gives
$\HS{\Gamma(it)}^2\leq1$ and $\HS{\Gamma(1/2+it)}^2\leq1$.
The maximum principle for bounded subharmonic functions on the strip
therefore gives $\HS{\Gamma(z)}^2\leq1$ throughout the strip.

Set $\omega=\one$ in \eqref{eq:h} and write
$h(z)=1-\Re\Tr\rho^{1/2}\Gamma(z)$.
By \eqref{eq:polarization} with $\omega=\one$, the quantity on the
left-hand side of \eqref{sv:complex-values} equals $2h(z)$.
To bound it by $\Delta$, we compare the Poisson kernels with
the weights in \eqref{eq:loss-expanded}. With corresponding choices of signs,
\begin{equation}\label{sv:interior-kernel-comparison}
 \frac{\sin(2\pi x)}{\cosh(2\pi(t-s))\mp\cos(2\pi x)}
 \leq\frac{2\bigl[\cosh(2\pi t)-\cos(2\pi x)\bigr]}
 {\sin(2\pi x)}\,
 \frac{1}{\cosh(2\pi s)\mp1}.
\end{equation}
For the minus signs, this comparison is made at $s\ne0$.
To verify both inequalities, write
$\theta=2\pi x$, $T=2\pi t$, $S=2\pi s$, and
$a=\cosh T-\cos\theta>0$ within this calculation. The elementary identities
\[
\begin{aligned}
 &2a\bigl[\cosh(T-S)-\cos\theta\bigr]
       -\sin^2\theta\bigl(\cosh S-1\bigr)=2\bigl[a\cosh(S/2)-\sinh T\sinh(S/2)\bigr]^2\geq0, \\
 &2a\bigl[\cosh(T-S)+\cos\theta\bigr]
       -\sin^2\theta\bigl(\cosh S+1\bigr)=2\bigl[a\sinh(S/2)-\sinh T\cosh(S/2)\bigr]^2\geq0
\end{aligned}
\]
follow by expanding the squares. Dividing by the positive denominators
proves \eqref{sv:interior-kernel-comparison}.

Applying the strip Poisson formula \eqref{eq:Poisson} to $h$ at
$z=x+it$ and using its nonnegative boundary values, the kernel
comparison \eqref{sv:interior-kernel-comparison} gives
\begin{equation}\label{eq:interior-poisson-bound}
\begin{aligned}
 2h(x+it)
 &=2\int_\R\left[
 \frac{\sin(2\pi x)h(is)}{\cosh(2\pi(t-s))-\cos(2\pi x)}
 +\frac{\sin(2\pi x)h(1/2+is)}{\cosh(2\pi(t-s))+\cos(2\pi x)}
 \right]\dd s\\
 &\leq\frac{4\bigl[\cosh(2\pi t)-\cos(2\pi x)\bigr]}{\sin(2\pi x)}
 \int_\R\left[
 \frac{h(is)}{\cosh(2\pi s)-1}
 +\frac{h(1/2+is)}{\cosh(2\pi s)+1}\right]\dd s\\
 &=\frac{2\bigl[\cosh(2\pi t)-\cos(2\pi x)\bigr]}
 {\pi\sin(2\pi x)}\,\Delta,
\end{aligned}
\end{equation}
The last equality follows from \eqref{eq:loss-expanded} and the
evenness of the lower-boundary weight.
The lower-boundary integrand is integrable at zero because
$h(is)=O(s^2)$, as established in \eqref{eq:h-derivative}.
Using \eqref{eq:polarization} with $\omega=\one$ proves
\eqref{sv:complex-values}.
\end{proof}

\begin{corollary}[Pointwise BS interior bound]\label{cor:bs-pointwise-interior}
Under the same assumptions, for $z=x+it$ with $0<x<1/2$,
\begin{equation}\label{eq:bs-pointwise-interior}
 0\leq\HS{\Gamma_{\mathrm{BS}}(0)-\Gamma_{\mathrm{BS}}(z)}^2
       +1-\HS{\Gamma_{\mathrm{BS}}(z)}^2
 \leq\frac{2[\cosh(2\pi t)-\cos(2\pi x)]}{\pi\sin(2\pi x)}
       \,\Delta_{\mathrm{BS}}.
\end{equation}
Both terms on the left are nonnegative. In particular,
\begin{equation}\label{eq:bs-interior-midpoint}
 \HS{\Gamma_{\mathrm{BS}}(0)-\Gamma_{\mathrm{BS}}(1/4)}^2
 \leq\frac{2\Delta_{\mathrm{BS}}}{\pi},
 \qquad
 \HS{\Gamma_{\mathrm{BS}}(1/4)}^2
 \geq1-\frac{2\Delta_{\mathrm{BS}}}{\pi}.
\end{equation}
\end{corollary}
\begin{proof}
The boundary contractions in Lemma~\ref{lem:bs-contraction} and the
subharmonic maximum principle give $\HS{\Gamma_{\mathrm{BS}}(z)}\leq1$
throughout the strip. Apply the proof of
Theorem~\ref{sv:complex-interior} to the nonnegative harmonic function
$h(z)=1-\Re\Tr[\Gamma_{\mathrm{BS}}(0)^\dagger\Gamma_{\mathrm{BS}}(z)]$,
using \eqref{eq:bs-reference-factorization} and \eqref{eq:bs-poisson}.
The same kernel comparison \eqref{sv:interior-kernel-comparison} gives
\eqref{eq:bs-pointwise-interior}, and $z=1/4$ gives the separate bounds.
\end{proof}

\subsection{Comparison with previous quantitative estimates}\label{sec:comparison}
We compare bounds on the recovery error at a prescribed rotation. 
For a channel that projects matrices onto a subalgebra $\mathcal N$
while preserving the trace (a trace-preserving conditional expectation),
Carlen--Vershynina~\cite[Corollary~1.8, Eq.~(1.33)]{CV} give
\begin{equation}\label{eq:CV-comparison}
 \|\rho-\mathcal R_0(\Phi(\rho))\|_1\le\frac8\pi
 \left(\|L_\sigma R_{\rho^{-1}}\|\,
       \|\rho_{\mathcal N}\|_\infty
       \|\sigma_{\mathcal N}^{-1}\|_\infty\right)^{1/2}\Delta^{1/4}.
\end{equation}
Here $L_\sigma R_{\rho^{-1}}$ is the relative modular operator on
Hilbert--Schmidt matrices. 

A stronger benchmark is from Gao--Wilde~\cite[Theorem~4.13 and Eq.~(61)]{GaoWilde}.
Writing $Q=\Tr\rho^2\sigma^{-1}$, their estimate implies, for
$0<\Delta<1$,
\begin{equation}\label{eq:GW-comparison}
 \|\rho-\mathcal R_s(\Phi(\rho))\|_1\le\frac{2\cosh(\pi s)}\pi\sqrt\Delta
 \left(4\sqrt Q+4+\sqrt{2\log(1/\Delta)}\right).
\end{equation}
This is obtained by choosing $S=\Delta$ and $T=\Delta^{-1}$ in their proof.
Their stated family also gives $O(\Delta^{1/2-\epsilon})$ for every fixed
$0<\epsilon<1/2$. Combining \eqref{eq:finite-recovery-simple} with
\eqref{eq:comparison-trace-recovery} gives a trace-norm recovery bound
proportional to $\sqrt{\Delta}$, without $Q$ or the factor
$\sqrt{\log(1/\Delta)}$.
Its coefficient is controlled either by the additive spectral count
$r_\rho+r_{\Phi(\sigma)}$ or by a logarithmic spectral bound involving
$\op{\log\sigma}$ and $\op{\log\Phi(\sigma)}$ instead of $Q$ which may potentially be penalized by $\sigma^{-1}$.

Theorem~\ref{thm:finite-recovery} controls the squared Hilbert--Schmidt difference together
with the correction term, which by \eqref{eq:comparison-trace-recovery} controls
the squared trace-norm recovery error.
In particular, the combined bound implies trace-norm recovery error
$O(\sqrt\Delta)$ for fixed $s$ and dimensions.

\begin{table}[tb]
\centering
\small
\renewcommand{\arraystretch}{1.24}
\setlength{\tabcolsep}{4pt}
\begin{tabular}{@{}p{0.23\textwidth}p{0.43\textwidth}p{0.30\textwidth}@{}}
\hline
Estimate & Error scaling & Scope of the constant\\
\hline
Carlen--Vershynina & $\|\rho-\mathcal R_0(\Phi(\rho))\|_1=O(\Delta^{1/4})$ & Spectral factors; unrotated map\\
Gao--Wilde & $\|\rho-\mathcal R_s(\Phi(\rho))\|_1=O(\sqrt{\Delta\log(1/\Delta)})$ & Prescribed $s$; $Q$-dependent\\
Theorem~\ref{thm:finite-recovery} (upper boundary) & $\HS{\rho^{1/2}-\Gamma(1/2+is)}^2+\varepsilon_1(s) =O(\Delta)$ & Prescribed $s$; fixed $r_\rho+r_{\Phi(\sigma)}$\\
Theorem~\ref{thm:finite-recovery} (lower boundary) & $\HS{\rho^{1/2}-\Gamma(-is)}^2+\varepsilon_0(s) =O(\Delta)$ & Prescribed $s$; fixed $r_\sigma+r_{\Phi(\rho)}$\\
Theorem~\ref{sv:complex-interior} (interior) & $\|\rho^{1/2}-\Gamma(z)\|_2^2+1-\|\Gamma(z)\|_2^2=O(\Delta)$ & Fixed $0<\Re z<1/2$; dimension- and spectrum-independent\\
\hline
\end{tabular}
\caption{Small-loss comparisons for fixed indicated parameters.
The real-parameter bounds more precisely use their respective boundary integrals.
The interior bound is from Section~\ref{sec:interior-pointwise}; its coefficient may diverge at the strip edges. 
}
\label{tab:pointwise-comparison}
\end{table}

\begin{proposition}[Optimality of the square-root exponent]\label{prop:optimal-exponent}
Neither prescribed Petz trace-distance recovery nor the general-channel
modular matrix difference $\HS{\rho^{1/2}-\Gamma(-is)}$ at fixed $s\ne0$ admits a bound
$C\Delta^\alpha$ with $\alpha>1/2$ and a constant uniform in fixed
dimensions.
\end{proposition}
\begin{proof}
For erasure to a scalar, take $\sigma=\one_2/2$ and
$\rho_\epsilon=\operatorname{diag}(1/2+\epsilon,1/2-\epsilon)$.
All rotations recover $\sigma$, while
$\|\rho-\mathcal R_s(\Phi(\rho))\|_1=2|\epsilon|$ and
$\Delta=2\epsilon^2+O(\epsilon^4)$.
For this same erasure example,
\[
 \HS{\rho^{1/2}-\Gamma(-is)}^2=
 \sum_{\pm}(1/2\pm\epsilon)
 \left|1-e^{is\log(1\pm2\epsilon)}\right|^2
 =4s^2\epsilon^2+O(\epsilon^4).
\]
Thus the general modular-error exponent is already optimal for a
scalar-output channel. This does not establish optimality of the
frequency-count prefactors.
These faithful, fixed-dimensional examples exclude $\alpha>1/2$.
\end{proof}

\section{Discussion}\label{sec:discussion}

The identities proved here express relative-entropy loss through exact,
nonnegative boundary and surface integrals. The boundary representation
composes state-reconstruction error and modular-flow mismatch, with
correction terms accounting for general channels. We discuss some possible extensions and applications below.

\subsection{Integral representations for other quantum divergence losses}
The Umegaki and BS identities share a strip Poisson--Green integral representation framework,
with rotated Petz recovery in the former and a generally nonlinear,
subnormalized reconstruction in the latter.
A natural question is whether analogous identities hold for Petz
R\'enyi~\cite{HiaiMosonyiPetzBeny2011}, sandwiched
R\'enyi~\cite{MullerLennert2013,WildeWinterYang2014}, and geometric
R\'enyi divergences~\cite{FangFawzi2021},
or for the recently introduced log-Euclidean R\'enyi conditional mutual information~\cite{rubboli2026logeuclideanrenyiconditionalmutual}.

\subsection{Modular flow and topological phases}
The vanishing of conditional mutual information, $I(X:Z\mid Y)_\rho = 0$, defines a quantum 
Markov chain~\cite{HaydenJozsaPetzWinter2004}. In our formulation (Corollary~\ref{cor:CMI}), 
the partial trace channel has a $*$-homomorphic adjoint, causing both correction terms to vanish 
identically. The CMI then splits cleanly into Petz recovery error and modular flow mismatch:
\begin{equation}\label{eq:disc-cmi}
  I(X:Z\mid Y)_\rho = \pi \int_\R \frac{\left\|\left(\rho^{is} - \rho_{XY}^{is}\rho_Y^{-is}\rho_{YZ}^{is}\right)\rho^{1/2}\right\|_2^2}{\cosh(2\pi s) - 1}\dd s 
                    + \pi \int_\R \frac{\left\|\rho^{1/2} - \rho_{YZ}^{1/2+is}\rho_Y^{-1/2-is}\rho_{XY}^{1/2+is}\rho^{-is}\right\|_2^2}{\cosh(2\pi s) + 1}\dd s.
\end{equation}

Conditional mutual information plays a central role in the study of topological phases of matter. 
To be more concrete, it has been used to define topological entanglement entropy~\cite{LevinWen2006}, which is a universal quantity characterizing the long-range entanglement in gapped quantum systems. 
It also plays a key role in the entanglement bootstrap approach where ones try to derive topological data and the theory of topological phases from some axioms that are expected to hold for gapped ground states~\cite{ShiKatoKim2020,Shi2020domainwall,Huang2021,Braiding2023}.
This approach has also been extended to mixed states~\cite{yang2025topologicalmixedstatesphases} where different partitions of CMI can be
identified as the absence of strong-to-weak symmetry breaking~\cite{Lee_2023,lee2022symmetry,lessa2024strongtoweakspontaneoussymmetrybreaking} and spatial uniformity. 
Moreover, CMI itself is a central tool for characterizing mixed-state phases
and phase transitions~\cite{sang2023mixedstate,Sang2024,sang2025mixedstatephaseslocalreversibility,Chen2026HierarchicalCMI,hsieh2026locallystable,chen2026notestrongquantummarkov}.

The exact decomposition \eqref{eq:disc-cmi} may help clarify how CMI
constrains phases of matter. One application concerns the modular
commutator,
\begin{equation}\label{eq:discussion-modular-commutator}
 J(X,Y,Z)_\rho
 :=i\Tr\rho_{XYZ}[K_{XY},K_{YZ}],
 \qquad K_R=-\log\rho_R.
\end{equation}
For suitable spatial tripartitions of two-dimensional gapped ground
states, $J$ has been proposed as a bulk probe of the chiral central
charge $c_-$~\cite{Kim2022c-minus,Modular-commutator-Gapped}. 
In Ref.~\cite{yang2026modularcommutatorrobusttopological}, we combine
finite-time Markov-decomposition error of modular flow implied by small CMI
with continuity estimates for $J$ to obtain a lower bound on
CMI when $J\ne0$. Under the entanglement area-law and geometric
assumptions of that work, with $\operatorname{dist}(A,C)$ comparable
to the geometric scale, the bound takes the schematic form
\[
 I(A:C\mid B)\gtrsim
 \exp\!\left[-O\!\left(\frac{\operatorname{dist}(A,C)}{|J|}\right)\right],
\]
up to polynomial prefactors, where $J$ is evaluated on the
corresponding partition. The first term of \eqref{eq:disc-cmi}
plays a key role in this argument, linking the present identity to constraints on CMI decay in
chiral states.

\subsection{Non-faithful states, infinite dimensions, and operator algebras}
The results in this paper are proved for faithful states. We believe that
their extension to non-faithful states in finite dimensions should be
straightforward under the support condition
$\supp\rho\subseteq\supp\sigma$. A detailed treatment of the support
conventions and limiting arguments needed for this extension is left to
future work.

The boundary identity in Theorem~\ref{thm:loss} follows from the strip
Poisson formula, the analytic overlap $f(z)$, and the Schwarz inequality
for the unital adjoint channel. These ingredients suggest investigating
an extension to faithful normal states on von Neumann algebras. Approximate recovery has
been established for general von Neumann subalgebras by Faulkner, Hollands,
Swingle, and Wang~\cite{FaulknerHollandsSwingleWang2022}, and for 2-positive
channels between general von Neumann algebras by Faulkner and
Hollands~\cite{FaulknerHollands2022}. These results, together with the universal
recovery inequalities of Junge and LaRacuente~\cite{JL}, provide relevant tools
for this direction. Establishing such an identity for relative entropy loss in general von Neumann algebras, including its scope in
quantum field theory, remains an open problem.
\section*{Acknowledgments}
T.H.Y is grateful to Bowen Shi for fruitful discussions on the modular flow Markov decomposition error. 
T.H.Y also thanks Marius Junge for helpful discussions and for drawing his attention to Ref.~\cite{JL}, and for helping him recognize the significance of the discovery made by ChatGPT.
T.H.Y is supported by Taiwan-UIUC fellowship program.

\appendix
\section{Nonnegativity of the boundary correction terms}\label{app:boundary-nonnegativity}
To see their nonnegativity, set
\begin{equation}\label{eq:UB}
 U_s=\Phi(\sigma)^{is}\Phi(\rho)^{-is},\qquad
 B_s=\Phi(\sigma)^{-1/2-is}\Phi(\rho)^{1/2+is}.
\end{equation}
Substituting the definitions of $\Gamma(z)$ and $\mathcal R_s$ gives
\begin{equation}\label{eq:eps0}
 \varepsilon_0(s) =\Tr\rho\bigl[\one_A-\Phi^\dagger(U_s)^*\Phi^\dagger(U_s)\bigr], \qquad \varepsilon_{1,\omega}(s) =\Tr\!\left\{\omega\sigma^{1/2+is} \bigl[\Phi^\dagger(B_sB_s^*)-\Phi^\dagger(B_s)\Phi^\dagger(B_s)^*\bigr] \sigma^{1/2-is}\right\}.
\end{equation}
Since $U_s$ is unitary and $\Phi^\dagger$ is unital and completely positive,
both quantities are nonnegative by the Schwarz inequality for $\Phi^\dagger$.
The same positive correction completes the product $\Gamma(1/2+is)\Gamma(1/2+is)^*$ to the rotated
Petz recovered state:
\begin{equation}\label{eq:recovery-schwarz-completion}
\begin{aligned}
 \mathcal R_s(\Phi(\rho))
 &=\Gamma(1/2+is)\Gamma(1/2+is)^* + \sigma^{1/2+is}
 \bigl[\Phi^\dagger(B_sB_s^*)
       -\Phi^\dagger(B_s)\Phi^\dagger(B_s)^*\bigr]
 \sigma^{1/2-is}.
\end{aligned}
\end{equation}
If $\Phi^\dagger$ is a $*$-homomorphism,
both $\varepsilon_0(s)$ and $\varepsilon_{1,\omega}(s)$ vanish.

\section{Proofs of the frequency-evaluation lemmas}\label{app:frequency-evaluation-proofs}

\renewcommand{\restatedlemmanumber}{\ref{lem:one-sided-evaluation}}
\begin{restatedlemma}[One-sided finite-frequency evaluation]
Let $\omega_1,\ldots,\omega_N$ be distinct real numbers and let
$p(t)=\sum_{j=1}^N c_je^{i\omega_jt}$, with $c_j\in\C$.
For every $a>0$,
\begin{equation*}
 |p(0)|^2\le aN\int_0^\infty e^{-at}|p(t)|^2\dd t.
\end{equation*}
The constant $aN$ is optimal for each such frequency set.
\end{restatedlemma}
\begin{proof}[Proof of Lemma~\ref{lem:one-sided-evaluation}]
Write $c=(c_1,\ldots,c_N)^{\mathsf T}$ and let
$\mathbf e=(1,\ldots,1)^{\mathsf T}\in\C^N$.
Expanding the squared modulus and integrating each term gives
\[
 \int_0^\infty e^{-at}|p(t)|^2\dd t=c^*Cc,
 \qquad
 C_{jk}=\int_0^\infty e^{-at}e^{i(\omega_k-\omega_j)t}\dd t
       =\frac{1}{a+i(\omega_j-\omega_k)}.
\]
The matrix $C$ is Hermitian positive definite. Indeed, if this integral
vanishes, continuity implies $p(t)=0$ for all $t\ge0$. Its first $N$
derivatives of orders $0,\ldots,N-1$ at zero then give
$\sum_j c_j(i\omega_j)^m=0$ for $m=0,\ldots,N-1$. The Vandermonde matrix is invertible
because the frequencies are distinct, so $c=0$.
Since $p(0)=\mathbf e^*c$, Cauchy--Schwarz yields
\[
 |p(0)|^2
 =\bigl|(C^{-1/2}\mathbf e)^*(C^{1/2}c)\bigr|^2
 \le (\mathbf e^*C^{-1}\mathbf e)(c^*Cc).
\]
It remains to compute $\mathbf e^*C^{-1}\mathbf e$.
The inverse-Cauchy identity (see \cite[Theorem~2.1]{GrinbergInverseCauchy}
and \cite{Schechter1959}) states that, for
$C_{jk}=(x_j+y_k)^{-1}$, this quantity is
$\sum_j x_j+\sum_k y_k$; here $x_j=a+i\omega_j$ and
$y_k=-i\omega_k$, so it equals $aN$. The following rational-interpolation
argument verifies this identity directly in our setting.
Put $r=C^{-1}\mathbf e$ and define
\[
 R(z)=\sum_{k=1}^N\frac{r_k}{z-i\omega_k},
 \qquad
 F(z)=1-\prod_{k=1}^N\frac{z-a-i\omega_k}{z-i\omega_k}.
\]
The equations $Cr=\mathbf e$ give $R(a+i\omega_j)=1$ for every $j$.
Also $F(a+i\omega_j)=1$, since one numerator factor vanishes.
Both rational functions have denominator
$Q(z)=\prod_k(z-i\omega_k)$ and numerator degree at most $N-1$;
for $F$, the leading terms of the two monic degree-$N$ polynomials
cancel. Thus $Q(z)(R(z)-F(z))$ has degree at most $N-1$ and vanishes
at the $N$ distinct points $a+i\omega_j$. It is identically zero,
so $R=F$.

As $z\to\infty$, their expansions are
\[
 R(z)=\frac{\sum_k r_k}{z}+O(z^{-2}),
 \qquad
 F(z)=1-\prod_{k=1}^N\left(1-\frac{a}{z}+O(z^{-2})\right)
     =\frac{aN}{z}+O(z^{-2}).
\]
Comparing coefficients gives
$\mathbf e^*C^{-1}\mathbf e=\sum_k r_k=aN$, proving
\eqref{eq:one-sided-evaluation}. Equality holds for
$c=C^{-1}\mathbf e$, so the constant is optimal.
\end{proof}

\renewcommand{\restatedlemmanumber}{\ref{lem:frequency-evaluation}}
\begin{restatedlemma}[Two finite-frequency evaluation inequalities]
Let $p(s)=\sum_{j=1}^N c_j e^{i\omega_js}$, where the $\omega_j$
are distinct real numbers and $c_j\in\C$. For every $s_0\in\R$,
\begin{equation*}
 |p(s_0)|^2\le 2N\cosh^2(\pi s_0)
 \int_\R\frac{\pi|p(s)|^2}{\cosh(2\pi s)+1}\dd s.
\end{equation*}
If additionally $p(0)=0$, then
\begin{equation*}
 |p(s_0)|^2\le N\sinh(2\pi|s_0|)
 \int_\R\frac{\pi|p(s)|^2}{\cosh(2\pi s)-1}\dd s.
\end{equation*}
\end{restatedlemma}
\begin{proof}[Proof of Lemma~\ref{lem:frequency-evaluation}]
Apply \eqref{eq:one-sided-evaluation} to $p(s_0+t)$ and $p(s_0-t)$ and
add the resulting inequalities. This gives
\begin{equation}\label{eq:two-sided-evaluation}
 |p(s_0)|^2\le\frac{aN}{2}\int_\R
                  e^{-a|s-s_0|}|p(s)|^2\dd s.
\end{equation}
For \eqref{eq:finite-frequency-evaluation}, $\left|\frac{\dd}{\dd s}\log\frac{\pi}{\cosh(2\pi s)+1}\right|\le2\pi$, so
$\frac{\pi}{\cosh(2\pi s)+1}\ge \frac{\pi}{\cosh(2\pi s_0)+1}e^{-2\pi|s-s_0|}$. Taking $a=2\pi$ in
\eqref{eq:two-sided-evaluation} gives
\[
 |p(s_0)|^2\le2N\cosh^2(\pi s_0)
 \int_\R\frac{\pi|p(s)|^2}{\cosh(2\pi s)+1}\dd s.
\]
For \eqref{eq:finite-frequency-vanishing}, suppose $s_0\ne0$ and set
$a=2\pi\coth(\pi|s_0|)$. Since $\pi/[\cosh(2\pi s)-1]$ is even and decreasing on $(0,\infty)$,
\[
 \frac{\pi}{\cosh(2\pi s)-1}\ge \frac{\pi}{\cosh(2\pi (|s_0|+|s-s_0|))-1}
           \ge \frac{\pi}{\cosh(2\pi s_0)-1}e^{-a|s-s_0|}.
\]
The second inequality uses
$-\frac{\dd}{\dd x}\log\frac{\pi}{\cosh(2\pi x)-1}=2\pi\coth(\pi x)\le a$ for $x\ge|s_0|$.
Equation~\eqref{eq:two-sided-evaluation} therefore yields
\[
 |p(s_0)|^2\le N\sinh(2\pi|s_0|)
 \int_\R\frac{\pi|p(s)|^2}{\cosh(2\pi s)-1}\dd s.
\]
When $p(0)=0$, the second estimate at zero is exact.
If $p$ takes values in a finite-dimensional normed space, apply the scalar
inequality to a norm-one linear functional attaining $\|p(s_0)\|$.
The integral of the scalar square is no larger than that of $\|p(s)\|^2$.
This proves the dimension-independent extension to vector and matrix
coefficients.

\end{proof}

\renewcommand{\restatedlemmanumber}{\ref{lem:pwlog-exponential}}
\begin{restatedlemma}[Two-sided exponential-weight evaluation]
Let
\begin{equation*}
 p(t)=c_0+\sum_{j=1}^{N}c_j e^{i\omega_jt},
 \qquad N\geq1,
\end{equation*}
where the $\omega_j$ are real and lie in an interval of length at most
$W\geq0$. There is no restriction on the location of this interval
relative to the constant frequency. For $a>0$, define
\begin{equation*}
 A_a(W,N)=\left[
 \sqrt{\frac{8a}{\pi^2}}+
 \sqrt{\frac{4W}{\pi}
       +\frac{4a}{\pi^2}\bigl(2\log(2N)+1\bigr)}
 \right]^2.
\end{equation*}
Then, for every $s_0\in\R$,
\begin{equation*}
 |p(s_0)|^2\leq A_a(W,N)
 \int_{\R} e^{-a|t-s_0|}|p(t)|^2\dd t.
\end{equation*}
\end{restatedlemma}
\begin{proof}[Proof of Lemma~\ref{lem:pwlog-exponential}]
It suffices first to consider scalar coefficients and $s_0=0$.
Repeated frequencies may be grouped. If a cluster frequency equals
zero, its contribution may either be grouped with $c_0$ or left in
the displayed representation; neither choice affects the argument.
Set
\begin{equation}
 R(z)=\frac{c_0}{z}+\sum_{j=1}^{N}\frac{c_j}{z-i\omega_j}.
 \label{eq:pwlog-resolvent}
\end{equation}
For $\operatorname{Re}z>0$ and $\operatorname{Re}z<0$, respectively,
\begin{equation}
 R(z)=\int_0^\infty e^{-zt}p(t)\dd t,
 \qquad
 R(z)=-\int_0^\infty e^{zt}p(-t)\dd t.
 \label{eq:pwlog-Laplace}
\end{equation}
Parseval's identity therefore gives
\begin{equation}
 \int_{\R}\!\left(
 |R(a/2+iy)|^2+|R(-a/2+iy)|^2\right)\dd y
 =2\pi\int_{\R}e^{-a|t|}|p(t)|^2\dd t.
 \label{eq:pwlog-Parseval}
\end{equation}

We construct an analytic interpolant on the vertical strip
$S_a=\{z:|\operatorname{Re}z|<a/2\}$. For a real $\omega$, put
\begin{equation}
 b_\omega(z)=-i\tan\!\left(\frac{\pi(z-i\omega)}{2a}\right).
 \label{eq:pwlog-factor}
\end{equation}
The function $b_\omega$ is analytic on a neighborhood of the closed
strip, vanishes at $i\omega$, and satisfies $|b_\omega|\leq1$ in
$\overline{S_a}$ and $|b_\omega|=1$ on its two boundary lines. More
explicitly, with $u=\pi(y-\omega)/a$,
\begin{equation}
 b_\omega(\pm a/2+iy)=\tanh u\mp i\operatorname{sech}u.
 \label{eq:pwlog-factor-boundary}
\end{equation}

For the cluster frequencies $\omega_1,\ldots,\omega_N$, define
\[
 B_1(z)=
 \begin{cases}
  \displaystyle\prod_{j=1}^{N}b_{\omega_j}(z),&N\text{ even},\\[4pt]
  \displaystyle b_{\omega_1}(z)\prod_{j=1}^{N}b_{\omega_j}(z),&N\text{ odd}.
 \end{cases}
\]
Thus, when $N$ is odd, the factor $b_{\omega_1}$ appears twice.
The total number of factors is $M=N$ for even $N$ and $M=N+1$ for
odd $N$; in both cases $M$ is even and $M\leq2N$.
Set
\begin{equation}
 B_0=b_0^2,\qquad H=1-B_0B_1.
 \label{eq:pwlog-interpolant}
\end{equation}
Thus $H(0)=H(i\omega_j)=1$ for all $j$. The even numbers of factors
ensure that $B_0$ and $B_1$ tend to $1$ at both vertical ends of the
strip. Consequently $H$ tends to zero exponentially at both ends.

For a function on the two boundary lines, write
\[
 \|F\|_{\partial S_a,2}^2
 =\int_{\R}\!\left(
 |F(a/2+iy)|^2+|F(-a/2+iy)|^2\right)\dd y.
\]
Since $H$ equals one at every pole of $R$, the sum of the residues of
$RH$ is $c_0+\sum_jc_j=p(0)$, with coincident poles grouped together.
Apply the residue theorem on the positively oriented rectangle
$\{z:|\Re z|\leq a/2,\ |\Im z|\leq T\}$, where
$T>\max_j|\omega_j|$, and let $T\to\infty$.
The horizontal contributions vanish because $R(x\pm iT)=O(T^{-1})$
and $H(x\pm iT)$ decays exponentially, uniformly for $|x|\leq a/2$.
Only the two vertical sides remain; parametrizing the right side upward
and the left side downward yields
\[
 p(0)=\frac{1}{2\pi}\int_{\R}\left[
 R(a/2+iy)H(a/2+iy)-R(-a/2+iy)H(-a/2+iy)
 \right]\dd y.
\]
Cauchy--Schwarz applied to both boundary lines together yields
\[
 |p(0)|^2\leq\frac{1}{4\pi^2}
 \|R\|_{\partial S_a,2}^2\|H\|_{\partial S_a,2}^2.
\]
Substituting the Parseval identity \eqref{eq:pwlog-Parseval} gives
\begin{equation}
 |p(0)|^2
 \leq\frac{\|H\|_{\partial S_a,2}^2}{2\pi}
 \int_{\R}e^{-a|t|}|p(t)|^2\dd t.
 \label{eq:pwlog-residue-estimate}
\end{equation}

We estimate the boundary norm of $H$ explicitly. First,
\eqref{eq:pwlog-factor-boundary} gives
\[
 |1-b_0(\pm a/2+iy)^2|^2
 =4\operatorname{sech}^2(\pi y/a),
\]
so
\begin{equation}
 \frac{\|1-B_0\|_{\partial S_a,2}^2}{2\pi}
 =\frac{8a}{\pi^2}.
 \label{eq:pwlog-singleton-norm}
\end{equation}
Let $[\alpha,\beta]$ contain the cluster frequencies, with
$\beta-\alpha\leq W$. Above this interval each factor differs from
$1$ by at most $2e^{-\pi(y-\beta)/a}$. Below it, each factor differs
from $-1$ by at most $2e^{-\pi(\alpha-y)/a}$. Since $M$ is even,
a telescoping product estimate and $|B_1|=1$ on the boundary give
\begin{equation}
 |1-B_1(\pm a/2+iy)|
 \leq2\min\!\left\{1,
 M e^{-\pi\operatorname{dist}(y,[\alpha,\beta])/a}\right\}.
 \label{eq:pwlog-tail}
\end{equation}
The tail integral on either side of the frequency interval is
controlled by
\[
 \int_0^\infty\min\{1,M^2e^{-2\pi u/a}\}\dd u
 =\frac{a}{\pi}\log M+\frac{a}{2\pi}.
\]
Counting both tails and both strip boundaries gives
\begin{equation}
 \frac{\|1-B_1\|_{\partial S_a,2}^2}{2\pi}
 \leq\frac{4W}{\pi}
       +\frac{4a}{\pi^2}(2\log M+1).
 \label{eq:pwlog-cluster-norm}
\end{equation}
Finally,
\[
 H=(1-B_0)+B_0(1-B_1),\qquad |B_0|=1
 \quad\hbox{on }\partial S_a.
\]
The triangle inequality in $L^2(\partial S_a)$, together with
\eqref{eq:pwlog-singleton-norm}, \eqref{eq:pwlog-cluster-norm} and
$M\leq2N$, proves
$\|H\|_{\partial S_a,2}^2/(2\pi)\leq A_a(W,N)$.
Substituting this bound into \eqref{eq:pwlog-residue-estimate}
proves \eqref{eq:pwlog-exponential} at $s_0=0$. Applying it to
$p(s_0+t)$ proves the result at an arbitrary $s_0$.
\end{proof}

\renewcommand{\restatedlemmanumber}{\ref{lem:pwlog-Poisson}}
\begin{restatedlemma}[Evaluation against the two Poisson weights]
Under the hypotheses of Lemma~\ref{lem:pwlog-exponential}, define
\begin{equation*}
 C_{\mathrm{eval}}(W,N)=\frac{8}{\pi^2}
 \left[2+\sqrt{W+4\log(2N)+2}\right]^2.
\end{equation*}
For every $s_0\in\R$,
\begin{equation*}
 |p(s_0)|^2\leq C_{\mathrm{eval}}(W,N)\cosh^2(\pi s_0)
 \int_{\R}\frac{\pi|p(t)|^2}{\cosh(2\pi t)+1}\dd t.
\end{equation*}
If in addition $p(0)=0$, then
\begin{align*}
 |p(s_0)|^2
 \leq\frac12 C_{\mathrm{eval}}(W,N)\sinh(2\pi|s_0|)
 \int_{\R}\frac{\pi|p(t)|^2}{\cosh(2\pi t)-1}\dd t.
\end{align*}
For the faithful density matrices $\sigma$ and $\Phi(\sigma)$,
\begin{equation*}
 C_{\mathrm{eval}}\!\left(w(\sigma)+w(\Phi(\sigma)),r_\sigma r_{\Phi(\sigma)}\right)
 \leq16(1+\op{\log\sigma}+\op{\log\Phi(\sigma)}).
\end{equation*}
\end{restatedlemma}
\begin{proof}[Proof of Lemma~\ref{lem:pwlog-Poisson}]
The logarithmic derivative of $\pi/[\cosh(2\pi t)+1]$ has
absolute value at most $2\pi$. Hence
\[
 \frac{\pi}{\cosh(2\pi t)+1}\geq \frac{\pi}{\cosh(2\pi s_0)+1}e^{-2\pi|t-s_0|}.
\]
Apply Lemma~\ref{lem:pwlog-exponential} with $a=2\pi$ and use
\[
 A_{2\pi}(W,N)=\frac4\pi
 \left[2+\sqrt{W+4\log(2N)+2}\right]^2,
 \qquad \frac{\cosh(2\pi s_0)+1}{\pi}=\frac{2\cosh^2(\pi s_0)}{\pi}.
\]
This gives \eqref{eq:pwlog-plus-evaluation}.

For the lower boundary, put $v=|s_0|>0$ and
$a_v=2\pi\coth(\pi v)$. The function
$\pi/[\cosh(2\pi t)-1]$ is even and decreasing on
$(0,\infty)$. For $x\geq v$,
\[
 -\frac{\dd}{\dd x}\log \frac{\pi}{\cosh(2\pi x)-1}
 =2\pi\coth(\pi x)\leq a_v.
\]
Since $|t|\leq v+|t-s_0|$, it follows that, for $t\ne0$,
\begin{equation}
 \frac{\pi}{\cosh(2\pi t)-1}\geq \frac{\pi}{\cosh(2\pi (v+|t-s_0|))-1}
 \geq \frac{\pi}{\cosh(2\pi v)-1}e^{-a_v|t-s_0|}.
 \label{eq:pwlog-minus-comparison}
\end{equation}
The vanishing assumption implies $p(t)=O(t)$ at zero, so the lower
weighted integral is finite there; both weighted integrals converge
at infinity. Lemma~\ref{lem:pwlog-exponential} with $a=a_v$ yields
the first inequality in \eqref{eq:pwlog-minus-evaluation}, because
\begin{align*}
 \frac{A_{a_v}(W,N)[\cosh(2\pi v)-1]}{\pi}
 &=\frac4{\pi^2}\sinh(2\pi v)
 \left[2+\sqrt{W\tanh(\pi v)+4\log(2N)+2}\right]^2\\
 &=\frac12 C_{\mathrm{eval}}\!\left(W\tanh(\pi v),N\right)\sinh(2\pi v).
\end{align*}
The second inequality uses $\tanh(\pi v)\leq1$ and monotonicity of
$C_{\mathrm{eval}}(W,N)$ in $W$. The Hilbert-space extension follows by summing
scalar estimates.

Finally, we prove \eqref{eq:pwlog-log-simplification}.
For every faithful density matrix $a$,
\begin{equation}
 w(a)\leq\op{\log a},\qquad
 \log r_a\leq\log\operatorname{rank}(a)
 \leq-\log\lambda_{\min}(a)=\op{\log a}.
 \label{eq:pwlog-normalization}
\end{equation}
The second inequality follows from
$1=\Tr a\geq\operatorname{rank}(a)\lambda_{\min}(a)$, while the first
uses $\lambda_{\max}(a)\leq1$. Thus the radicand in
\eqref{eq:pwlog-K} is at most $5\op{\log\sigma}+5\op{\log\Phi(\sigma)}+4\log2+2$.
Applying $(2+\sqrt{x})^2\leq8+2x$, followed by
$5<\pi^2$ and $6+4\log2<\pi^2$, yields
\eqref{eq:pwlog-log-simplification}.
\end{proof}

\section{Proofs of the BS identities and further consequences}
\label{sec:bs-poisson-bures}
We use the notation of Sections~\ref{sec:bs-results}
and~\ref{sec:bs-fidelity-variational}.
The two expressions in \eqref{eq:bs-definition} agree by polar decomposition, or by
$T^\dagger\log(TT^\dagger)T=(T^\dagger T)\log(T^\dagger T)$ with
$T=\rho^{1/2}\sigma^{-1/2}$.

For the appendix calculations, abbreviate
\begin{equation}\label{eq:bs-output-family}
 Z_{\mathrm{BS}}(z)
 :=\Phi(\sigma)^{-1/2}\widehat\ell_{\mathrm{BS}}^{\,z}
                  \Phi(\sigma)^{1/2}.
\end{equation}
Thus $\Gamma_{\mathrm{BS}}(z)=\ell_{\mathrm{BS}}^{1/2-z}
\sigma^{1/2}\Phi^\dagger(Z_{\mathrm{BS}}(z))$.

\subsection{Weighted Schwarz contraction}
\begin{lemma}[Weighted Schwarz contraction]\label{lem:bs-contraction}
Both defects in \eqref{eq:bs-defects} are nonnegative. More explicitly,
\begin{align}
 \delta_{\mathrm{BS},0}(s)
 &=\Tr\rho\Bigl[
   \Phi^\dagger\bigl(Z_{\mathrm{BS}}(is)Z_{\mathrm{BS}}(is)^\dagger\bigr)
   -\Phi^\dagger\bigl(Z_{\mathrm{BS}}(is)\bigr)
    \Phi^\dagger\bigl(Z_{\mathrm{BS}}(is)\bigr)^\dagger
   \Bigr],\label{eq:bs-schwarz-lower}\\
 \delta_{\mathrm{BS},1}(s)
 &=\Tr\sigma\Bigl[
   \Phi^\dagger\bigl(Z_{\mathrm{BS}}(1/2+is)
                         Z_{\mathrm{BS}}(1/2+is)^\dagger\bigr)
   -\Phi^\dagger\bigl(Z_{\mathrm{BS}}(1/2+is)\bigr)
    \Phi^\dagger\bigl(Z_{\mathrm{BS}}(1/2+is)\bigr)^\dagger
   \Bigr].\label{eq:bs-schwarz-upper}
\end{align}
\end{lemma}

\begin{proof}
The unital completely positive map $\Phi^\dagger$ satisfies
\begin{equation}\label{eq:bs-schwarz}
 \Phi^\dagger(M)\Phi^\dagger(M)^\dagger
 \leq\Phi^\dagger(MM^\dagger)
\end{equation}
for every output matrix $M$. Cyclicity of trace gives
\begin{align}
 \Tr\Phi(\rho)Z_{\mathrm{BS}}(is)Z_{\mathrm{BS}}(is)^\dagger
 &=\Tr\widehat\ell_{\mathrm{BS}}\,
       \widehat\ell_{\mathrm{BS}}^{is}\Phi(\sigma)
       \widehat\ell_{\mathrm{BS}}^{-is}=1,
 \label{eq:bs-lower-unit-weight}\\
 \Tr\Phi(\sigma)Z_{\mathrm{BS}}(1/2+is)
                       Z_{\mathrm{BS}}(1/2+is)^\dagger
 &=\Tr\widehat\ell_{\mathrm{BS}}^{1/2+is}\Phi(\sigma)
       \widehat\ell_{\mathrm{BS}}^{1/2-is}=1.
 \label{eq:bs-upper-unit-weight}
\end{align}
Here $\Tr\Phi(\sigma)\widehat\ell_{\mathrm{BS}}=\Tr\Phi(\rho)=1$.
On the other hand, the factors $\ell_{\mathrm{BS}}^{-is}$ in the
boundary values of $\Gamma_{\mathrm{BS}}$ are unitary, so
\begin{equation*}
 \HS{\Gamma_{\mathrm{BS}}(is)}^2 =\Tr\rho\,\Phi^\dagger\bigl(Z_{\mathrm{BS}}(is)\bigr) \Phi^\dagger\bigl(Z_{\mathrm{BS}}(is)\bigr)^\dagger,\qquad \HS{\Gamma_{\mathrm{BS}}(1/2+is)}^2 =\Tr\sigma\,\Phi^\dagger\bigl(Z_{\mathrm{BS}}(1/2+is)\bigr) \Phi^\dagger\bigl(Z_{\mathrm{BS}}(1/2+is)\bigr)^\dagger.
\end{equation*}
Subtracting proves \eqref{eq:bs-schwarz-lower}--\eqref{eq:bs-schwarz-upper};
\eqref{eq:bs-schwarz} proves positivity. Notice that
$Z_{\mathrm{BS}}(is)$ need not be unitary. It is the weighted trace
identities above, rather than unitarity of this operator, that give the
contractions.
\end{proof}

\subsection{Proof of the boundary loss identity}
\begin{proof}[Proof of the boundary identity in Theorem~\ref{thm:bs-poisson}]
The function
\[
 h(z)=1-\Re\Tr\!\left[\Gamma_{\mathrm{BS}}(0)^\dagger
                              \Gamma_{\mathrm{BS}}(z)\right]
\]
is bounded and harmonic on the strip, with $h(0)=0$ by
\eqref{eq:bs-reference-factorization}. Differentiating
\eqref{eq:bs-Gamma} at zero and using cyclicity gives
\[
 \partial_x h(0)
 =\Tr\sigma\ell_{\mathrm{BS}}\log\ell_{\mathrm{BS}}
 -\Tr\Phi(\sigma)\widehat\ell_{\mathrm{BS}}\log\widehat\ell_{\mathrm{BS}}
 =\Delta_{\mathrm{BS}}.
\]
The derivative of the analytic overlap at zero is real, so
$h(is)=O(s^2)$. For $j\in\{0,1\}$, expanding the squared norm and
using \eqref{eq:bs-defects} yields
\[
 2h(j/2+is)
 =\HS{\Gamma_{\mathrm{BS}}(0)-\Gamma_{\mathrm{BS}}(j/2+is)}^2
 +\delta_{\mathrm{BS},j}(s).
\]
Lemma~\ref{lem:Poisson-derivative} now gives \eqref{eq:bs-poisson}.
Nonnegativity follows from Lemma~\ref{lem:bs-contraction}; convergence
follows from $h(is)=O(s^2)$ and exponential decay of the kernels.
\end{proof}

\subsection{Proof of the fidelity remainder}
\begin{proof}[Proof of Theorem~\ref{thm:bs-fidelity-remainder}]
The averaging weight has total mass one:
\begin{equation}\label{eq:bs-boundary-probability}
 \int_{\R}\frac{\pi}{\cosh(2\pi s)+1}\dd s
 =\frac12\int_{\R}\pi\,\operatorname{sech}^2(\pi s)\dd s=1.
\end{equation}
Together with \eqref{eq:bs-subnormalization}, this proves positivity and
subnormalization of \eqref{eq:bs-averaged-reconstruction}. Moreover,
\begin{equation}\label{eq:bs-integrated-trace-defect}
 1-\Tr\overline\tau_{\mathrm{BS}}
 =\int_{\R}\frac{\pi\,\delta_{\mathrm{BS},1}(s)}
                     {\cosh(2\pi s)+1}\dd s.
\end{equation}
Thus the missing trace is exactly the integrated Schwarz defect.

Both $\Tr(\rho\omega)$ and $\Tr(\overline\tau_{\mathrm{BS}}\omega^{-1})$ are positive: the averaged reconstruction is nonzero,
since $\Tr\sigma\Phi^\dagger(Z_{\mathrm{BS}}(1/2))
=\Tr\Phi(\sigma)\widehat\ell_{\mathrm{BS}}^{1/2}>0$,
and $\Gamma_{\mathrm{BS}}(1/2+is)$ is continuous in $s$.

For fixed $\omega>0$, start with the bounded analytic function on
$0\leq\Re z\leq1/2$,
\begin{equation}\label{eq:bs-scalar-function}
 f_{\mathrm{BS},\omega}(z)
 :=\bigl[\Tr(\rho\omega)\Tr(\overline\tau_{\mathrm{BS}}\omega^{-1})\bigr]^{-z}
 \Tr\!\left[\Gamma_{\mathrm{BS}}(0)^\dagger\Gamma_{\mathrm{BS}}(z)\right].
\end{equation}
Boundedness on the closed strip follows from finite dimensionality,
unitarity of the imaginary powers, and $\Tr(\rho\omega)\Tr(\overline\tau_{\mathrm{BS}}\omega^{-1})>0$.
We have $f_{\mathrm{BS},\omega}(0)=1$, and differentiation gives
\begin{equation}\label{eq:bs-scalar-derivative}
 \begin{aligned}
 f_{\mathrm{BS},\omega}'(0)
 &=-\log\bigl[\Tr(\rho\omega)\Tr(\overline\tau_{\mathrm{BS}}\omega^{-1})\bigr]
 -\Tr\sigma\ell_{\mathrm{BS}}\log\ell_{\mathrm{BS}}
 +\Tr\Phi(\rho)\Phi(\sigma)^{-1/2}\log\widehat\ell_{\mathrm{BS}}\,\Phi(\sigma)^{1/2}\\
 &=-\log\bigl[\Tr(\rho\omega)\Tr(\overline\tau_{\mathrm{BS}}\omega^{-1})\bigr]-\Delta_{\mathrm{BS}}\in\R.
 \end{aligned}
\end{equation}
The last equality uses
$\Phi(\rho)=\Phi(\sigma)^{1/2}\widehat\ell_{\mathrm{BS}}\Phi(\sigma)^{1/2}$
and cyclicity, without commuting distinct factors.
Thus $1-\Re f_{\mathrm{BS},\omega}(is)=O(s^2)$ at zero.
The boundary overlaps are
\begin{equation}\label{eq:bs-boundary-overlaps}
 f_{\mathrm{BS},\omega}(j/2+is)
 =\bigl[\Tr(\rho\omega)\Tr(\overline\tau_{\mathrm{BS}}\omega^{-1})\bigr]^{-j/2-is}
   \Tr\!\left[\Gamma_{\mathrm{BS}}(0)^\dagger\Gamma_{\mathrm{BS}}(j/2+is)\right],
 \qquad j\in\{0,1\}.
\end{equation}
Applying the differentiated Poisson formula on the strip gives
\begin{equation}
 \Delta_{\mathrm{BS}}+\log\bigl[\Tr(\rho\omega)\Tr(\overline\tau_{\mathrm{BS}}\omega^{-1})\bigr]
 =
 2\pi\int_{\R}\Bigg[
 \frac{1-\Re f_{\mathrm{BS},\omega}(is)}
      {\cosh(2\pi s)-1}
 +
 \frac{1-\Re f_{\mathrm{BS},\omega}(1/2+is)}
      {\cosh(2\pi s)+1}
 \Bigg]\dd s .
 \label{eq:bs-fid-poisson-tilted}
\end{equation}

At $z=is$,
\begin{equation}
 2\bigl(1-\Re f_{\mathrm{BS},\omega}(is)\bigr)
 =
 \|\Gamma_{\mathrm{BS}}(0)-\bigl[\Tr(\rho\omega)\Tr(\overline\tau_{\mathrm{BS}}\omega^{-1})\bigr]^{-is}\Gamma_{\mathrm{BS}}(is)\|_2^2
 +\delta_{\mathrm{BS},0}(s).
 \label{eq:bs-fid-lower-polarization}
\end{equation}
At $z=1/2+is$, define
\begin{equation}
 u_{\omega}
 :=
 \frac{\Gamma_{\mathrm{BS}}(0)\omega^{1/2}}{\sqrt{\Tr(\rho\omega)}},
 \qquad
 v_{\omega,s}
 :=
 \bigl[\Tr(\rho\omega)\Tr(\overline\tau_{\mathrm{BS}}\omega^{-1})\bigr]^{-is}
 \frac{\Gamma_{\mathrm{BS}}(1/2+is)\omega^{-1/2}}{\sqrt{\Tr(\overline\tau_{\mathrm{BS}}\omega^{-1})}}.
\end{equation}
Then $\|u_{\omega}\|_2=1$ and, by cyclicity,
\begin{equation}
 \Tr(u_{\omega}^\dagger v_{\omega,s})
 =
 f_{\mathrm{BS},\omega}(1/2+is).
\end{equation}
Hence
\begin{equation}
 2\bigl(1-\Re f_{\mathrm{BS},\omega}(1/2+is)\bigr)
 =
 \|u_{\omega}-v_{\omega,s}\|_2^2
 +
 1-\frac{
   \Tr\bigl(\tau_{\mathrm{BS}}(s)\omega^{-1}\bigr)
 }{\Tr(\overline\tau_{\mathrm{BS}}\omega^{-1})}.
 \label{eq:bs-fid-upper-polarization}
\end{equation}
The final term integrates to zero against $\pi/[\cosh(2\pi s)+1]$
by \eqref{eq:bs-boundary-probability}, \eqref{eq:bs-averaged-reconstruction},
and linearity of the trace.
Substitution into \eqref{eq:bs-fid-poisson-tilted} proves
\eqref{eq:bs-fid-test-identity}.
All integrals converge absolutely: the nonnegative lower-boundary
numerator is real analytic and $O(s^2)$ by
\eqref{eq:bs-fid-lower-polarization}, so its apparent singularity is
removable. The upper-boundary denominator is at least two, all matrix families
are bounded for fixed $\omega$, and both kernels decay exponentially.

Finally, Alberti's variational characterization of Uhlmann fidelity~\cite{Alberti1983},
valid for positive operators without trace normalization, is
\begin{equation}
 F(\rho,\overline\tau_{\mathrm{BS}})
 =
 \inf_{\omega>0}
 \Tr(\rho\omega)\,
 \Tr(\overline\tau_{\mathrm{BS}}\omega^{-1}).
 \label{eq:bs-fid-alberti}
\end{equation}
Since the left-hand side of
\eqref{eq:bs-fid-test-identity} is independent of $\omega$,
taking the supremum of the first term is equivalent to taking the
infimum of the remainder, and yields
\eqref{eq:bs-fid-remainder}.
\end{proof}

\section{Green-function boundary limit}
\subsection{Passing the limit through the surface integral}\label{app:green-boundary-limit}
For the Green function in \eqref{sv:Green}, we establish the $L^1$
convergence used in the proof of Lemma~\ref{sv:greenlemma}.
We compare the limit of the integrals with the integral of the
pointwise limit. First, integration in $t$ gives
\[
 \int_\R G(r,x+it)\dd t=
 \begin{cases}
  2x(1/2-r),&0\le x\le r,\\
  2r(1/2-x),&r\le x\le1/2.
 \end{cases}
\]
Indeed, integrating the defining equation for $G$ in $t$ shows that
this integral has negative second derivative $\delta(x-r)$ and
vanishes at $x=0,1/2$. It is therefore continuous and linear on each
side of $r$, with derivative jump $-1$, which determines the displayed
expression. The $t$-derivative term integrates to zero by exponential
decay, with the equation understood distributionally. Consequently,
\[
\begin{aligned}
 \lim_{r\downarrow0}\int_0^{1/2}\!\int_\R
 \frac{G(r,x+it)}r\dd t\dd x
 &=\lim_{r\downarrow0}\left[
 \frac{2(1/2-r)}r\int_0^r x\dd x
 +2\int_r^{1/2}(1/2-x)\dd x\right]\\
 &=\lim_{r\downarrow0}\frac{1/2-r}{2}=\frac14.
\end{aligned}
\]
On the other hand, \eqref{sv:Green} gives the pointwise limit
$G(r,x+it)/r\to\sin(2\pi x)/[\cosh(2\pi t)-\cos(2\pi x)]$.
Integrating this limit instead gives
\[
 \int_0^{1/2}\!\int_\R
 \frac{\sin(2\pi x)}{\cosh(2\pi t)-\cos(2\pi x)}\dd t\dd x
 =\int_0^{1/2}(1-2x)\dd x=\frac14.
\]
Thus the two calculations agree. Since $G/r$ and its pointwise limit
are nonnegative, convergence of their integrals implies $L^1$
convergence by Scheff\'e's lemma.
Multiplication by the bounded function $\|F'(x+it)\|^2$ preserves
convergence of the integrals. This permits passing the limit through
the surface integral in Green's representation divided by $r$.

\bibliographystyle{apsrev4-2}
\makeatletter
\immediate\write\@auxout{\string\citation{IMFPetzBibliographyControl}}
\makeatother
\let\OriginalBibitemShut\BibitemShut
\renewcommand*{\BibitemShut}[1]{\OriginalBibitemShut{Stop}\backrefprint}
\bibliography{references}

@CONTROL{IMFPetzBibliographyControl,title="1",pages="1"}

@misc{yang2025topologicalmixedstatesphases,
      title={Topological Mixed States: Phases of Matter from Axiomatic Approaches}, 
      author={Tai-Hsuan Yang and Bowen Shi and Jong Yeon Lee},
      year={2025},
      eprint={2506.04221},
      archivePrefix={arXiv},
      primaryClass={cond-mat.str-el},
      url={https://arxiv.org/abs/2506.04221}, 
}

@article{lessa2024strongtoweakspontaneoussymmetrybreaking,
  title = {Strong-to-Weak Spontaneous Symmetry Breaking in Mixed Quantum States},
  author = {Lessa, Leonardo A. and Ma, Ruochen and Zhang, Jian-Hao and Bi, Zhen and Cheng, Meng and Wang, Chong},
  journal = {PRX Quantum},
  volume = {6},
  pages = {010344},
  year = {2025},
  doi = {10.1103/PRXQuantum.6.010344},
  eprint = {2405.03639},
  archivePrefix = {arXiv},
  primaryClass = {quant-ph}
}

@article{lee2022symmetry,
  title = {Symmetry protected topological phases under decoherence},
  author = {Lee, Jong Yeon and You, Yi-Zhuang and Xu, Cenke},
  journal = {Quantum},
  volume = {9},
  pages = {1607},
  year = {2025},
  doi = {10.22331/q-2025-01-23-1607},
  eprint = {2210.16323},
  archivePrefix = {arXiv},
  primaryClass = {cond-mat.str-el}
}

@article{Lee_2023,
   title={Quantum Criticality Under Decoherence or Weak Measurement},
   volume={4},
   pages={030317},
   ISSN={2691-3399},
   url={http://dx.doi.org/10.1103/PRXQuantum.4.030317},
   DOI={10.1103/prxquantum.4.030317},
   number={3},
   journal={PRX Quantum},
   publisher={American Physical Society (APS)},
   author={Lee, Jong Yeon and Jian, Chao-Ming and Xu, Cenke},
   year={2023},
   month=aug 
}

@article{sang2023mixedstate,
  title = {Mixed-State Quantum Phases: Renormalization and Quantum Error Correction},
  author = {Sang, Shengqi and Zou, Yijian and Hsieh, Timothy H.},
  journal = {Phys. Rev. X},
  volume = {14},
  issue = {3},
  pages = {031044},
  numpages = {24},
  year = {2024},
  month = {Sep},
  publisher = {American Physical Society},
  doi = {10.1103/PhysRevX.14.031044},
  url = {https://link.aps.org/doi/10.1103/PhysRevX.14.031044}
}

@article{LevinWen2006,
  title = {Detecting Topological Order in a Ground State Wave Function},
  author = {Levin, Michael and Wen, Xiao-Gang},
  journal = {Phys. Rev. Lett.},
  volume = {96},
  issue = {11},
  pages = {110405},
  numpages = {4},
  year = {2006},
  month = {Mar},
  publisher = {American Physical Society},
  doi = {10.1103/PhysRevLett.96.110405},
  url = {https://link.aps.org/doi/10.1103/PhysRevLett.96.110405}
}

@article{Sang2024,
  title = {Stability of Mixed-State Quantum Phases via Finite Markov Length},
  author = {Sang, Shengqi and Hsieh, Timothy H.},
  journal = {Phys. Rev. Lett.},
  volume = {134},
  issue = {7},
  pages = {070403},
  numpages = {7},
  year = {2025},
  month = {Feb},
  publisher = {American Physical Society},
  doi = {10.1103/PhysRevLett.134.070403},
  url = {https://link.aps.org/doi/10.1103/PhysRevLett.134.070403}
}

@misc{sang2025mixedstatephaseslocalreversibility,
      title={Mixed-state phases from local reversibility}, 
      author={Shengqi Sang and Leonardo A. Lessa and Roger S. K. Mong and Tarun Grover and Chong Wang and Timothy H. Hsieh},
      year={2025},
      eprint={2507.02292},
      archivePrefix={arXiv},
      primaryClass={quant-ph},
      url={https://arxiv.org/abs/2507.02292}, 
}

@misc{hsieh2026locallystable,
      title={A Unified Framework for Locally Stable Phases}, 
      author={Zhi Li and Raz Firanko and Timothy H. Hsieh},
      year={2026},
      eprint={2605.00088},
      archivePrefix={arXiv},
      primaryClass={quant-ph},
      url={https://arxiv.org/abs/2605.00088}, 
}

@misc{chen2026notestrongquantummarkov,
      title={Note on Strong Quantum Markov Properties}, 
      author={Chi-Fang Chen},
      year={2026},
      eprint={2605.02877},
      archivePrefix={arXiv},
      primaryClass={quant-ph},
      url={https://arxiv.org/abs/2605.02877}, 
}

@ARTICLE{Braiding2023, 
       author = {{Shi}, Bowen and {Huang}, Jin-Long and {McGreevy}, John},
        title = "{Remote detectability from entanglement bootstrap I: Kirby's torus trick}",
      journal = {SciPost Physics},
         year = 2025,
        month = apr,
       volume = {18},
       number = {4},
          eid = {126},
        pages = {126},
          doi = {10.21468/SciPostPhys.18.4.126},
archivePrefix = {arXiv},
       eprint = {2301.07119},
 primaryClass = {quant-ph},
       adsurl = {https://ui.adsabs.harvard.edu/abs/2025ScPP...18..126S}
}

@ARTICLE{Huang2021,
       author = {{Huang}, Jin-Long and {McGreevy}, John and {Shi}, Bowen},
        title = "{Knots and entanglement}",
      journal = {SciPost Physics},
         year = 2023,
        month = jun,
       volume = {14},
       number = {6},
          eid = {141},
        pages = {141},
          doi = {10.21468/SciPostPhys.14.6.141},
archivePrefix = {arXiv},
       eprint = {2112.08398},
 primaryClass = {hep-th},
       adsurl = {https://ui.adsabs.harvard.edu/abs/2023ScPP...14..141H}
}

@ARTICLE{Shi2020domainwall,
       author = {{Shi}, Bowen and {Kim}, Isaac H.},
        title = "{Entanglement bootstrap approach for gapped domain walls}",
      journal = {\prb},
         year = 2021,
        month = mar,
       volume = {103},
       number = {11},
          eid = {115150},
        pages = {115150},
          doi = {10.1103/PhysRevB.103.115150},
archivePrefix = {arXiv},
       eprint = {2008.11793},
 primaryClass = {cond-mat.str-el},
       adsurl = {https://ui.adsabs.harvard.edu/abs/2021PhRvB.103k5150S}
}

@article{lieb1973proof,
  author = {Lieb, Elliott H. and Ruskai, Mary Beth},
  title = {Proof of the Strong Subadditivity of Quantum-Mechanical Entropy},
  journal = {Journal of Mathematical Physics},
  volume = {14},
  number = {12},
  pages = {1938--1941},
  year = {1973},
  doi = {10.1063/1.1666274}
}

@ARTICLE{Kim2022c-minus,
       author = {{Kim}, Isaac H. and {Shi}, Bowen and {Kato}, Kohtaro and {Albert}, Victor V.},
        title = "{Chiral Central Charge from a Single Bulk Wave Function}",
      journal = {Physical Review Letters},
         year = 2022,
        month = apr,
       volume = {128},
       number = {17},
          eid = {176402},
        pages = {176402},
          doi = {10.1103/PhysRevLett.128.176402},
          url = {https://doi.org/10.1103/PhysRevLett.128.176402},
archivePrefix = {arXiv},
       eprint = {2110.06932},
 primaryClass = {quant-ph},
       adsurl = {https://ui.adsabs.harvard.edu/abs/2022PhRvL.128q6402K}
}

@ARTICLE{Fawzi2014,
       author = {{Fawzi}, Omar and {Renner}, Renato},
        title = "{Quantum Conditional Mutual Information and Approximate Markov Chains}",
      journal = {Communications in Mathematical Physics},
         year = 2015,
        month = dec,
       volume = {340},
       number = {2},
        pages = {575-611},
          doi = {10.1007/s00220-015-2466-x},
archivePrefix = {arXiv},
       eprint = {1410.0664},
 primaryClass = {quant-ph},
       adsurl = {https://ui.adsabs.harvard.edu/abs/2015CMaPh.340..575F}
}

@article{Junge2015,
  author = {Junge, Marius and Renner, Renato and Sutter, David and Wilde, Mark M. and Winter, Andreas},
  title = {Universal Recovery Maps and Approximate Sufficiency of Quantum Relative Entropy},
  journal = {Annales Henri Poincare},
  volume = {19},
  number = {10},
  pages = {2955--2978},
  year = {2018},
  doi = {10.1007/s00023-018-0716-0},
  eprint = {1509.07127},
  archivePrefix = {arXiv},
  primaryClass = {quant-ph}
}

@article{pinchedpetz2016,
   title={Strengthened Monotonicity of Relative Entropy via Pinched Petz Recovery Map},
   volume={62},
   ISSN={1557-9654},
   url={http://dx.doi.org/10.1109/TIT.2016.2545680},
   DOI={10.1109/tit.2016.2545680},
   number={5},
   journal={IEEE Transactions on Information Theory},
   publisher={Institute of Electrical and Electronics Engineers (IEEE)},
   author={Sutter, David and Tomamichel, Marco and Harrow, Aram W.},
   year={2016},
   month=may, pages={2907–2913} }

@ARTICLE{Modular-commutator-Gapped,
       author = {{Kim}, Isaac H. and {Shi}, Bowen and {Kato}, Kohtaro and {Albert}, Victor V.},
        title = "{Modular commutator in gapped quantum many-body systems}",
      journal = {Physical Review B},
         year = 2022,
        month = aug,
       volume = {106},
       number = {7},
          eid = {075147},
        pages = {075147},
          doi = {10.1103/PhysRevB.106.075147},
          url = {https://doi.org/10.1103/PhysRevB.106.075147},
archivePrefix = {arXiv},
       eprint = {2110.10400},
 primaryClass = {quant-ph},
       adsurl = {https://ui.adsabs.harvard.edu/abs/2022PhRvB.106g5147K}
}

@article{ShiKatoKim2020,
  author = {Shi, Bowen and Kato, Kohtaro and Kim, Isaac H.},
  title = {Fusion Rules from Entanglement},
  journal = {Annals of Physics},
  volume = {418},
  pages = {168164},
  year = {2020},
  doi = {10.1016/j.aop.2020.168164},
  url = {https://doi.org/10.1016/j.aop.2020.168164}
}

@article{HaydenJozsaPetzWinter2004,
  author = {Hayden, Patrick and Jozsa, Richard and Petz, D{\'e}nes and Winter, Andreas},
  title = {Structure of States Which Satisfy Strong Subadditivity of Quantum Entropy with Equality},
  journal = {Communications in Mathematical Physics},
  volume = {246},
  number = {2},
  pages = {359--374},
  year = {2004},
  doi = {10.1007/s00220-004-1049-z},
  url = {https://doi.org/10.1007/s00220-004-1049-z}
}

@article{SutterBertaTomamichel2017,
  author = {Sutter, David and Berta, Mario and Tomamichel, Marco},
  title = {Multivariate Trace Inequalities},
  journal = {Communications in Mathematical Physics},
  volume = {352},
  number = {1},
  pages = {37--58},
  year = {2017},
  doi = {10.1007/s00220-016-2778-5},
  url = {https://doi.org/10.1007/s00220-016-2778-5}
}

@article{widder1961functions,
  author  = {Widder, D. V.},
  title   = {Functions Harmonic in a Strip},
  journal = {Proceedings of the American Mathematical Society},
  volume  = {12},
  number  = {1},
  pages   = {67--72},
  year    = {1961},
  doi     = {10.2307/2034126}
}

@misc{yang2026modularcommutatorrobusttopological,
  title         = {Modular commutator as a robust topological invariant and approximate Markovianity},
  author        = {Tai Hsuan Yang},
  year          = {2026},
  eprint        = {2609.09019},
  archivePrefix = {arXiv},
  primaryClass  = {quant-ph},
  url           = {https://arxiv.org/abs/2609.09019}
}

@article{BFT,
  author = {Berta, Mario and Fawzi, Omar and Tomamichel, Marco},
  title = {On Variational Expressions for Quantum Relative Entropies},
  journal = {Letters in Mathematical Physics},
  volume = {107},
  number = {12},
  pages = {2239--2265},
  year = {2017},
  doi = {10.1007/s11005-017-0990-7}
}

@article{CV,
  author = {Carlen, Eric A. and Vershynina, Anna},
  title = {Recovery Map Stability for the Data Processing Inequality},
  journal = {Journal of Physics A: Mathematical and Theoretical},
  volume = {53},
  number = {3},
  pages = {035204},
  year = {2020},
  doi = {10.1088/1751-8121/ab5ab7}
}

@article{JL,
  author = {Junge, Marius and LaRacuente, Nicholas},
  title = {Multivariate Trace Inequalities, {$p$}-Fidelity, and Universal Recovery Beyond Tracial Settings},
  journal = {Journal of Mathematical Physics},
  volume = {63},
  number = {12},
  pages = {122204},
  year = {2022},
  doi = {10.1063/5.0066653}
}

@article{Lindblad75,
  author = {Lindblad, G{\"o}ran},
  title = {Completely Positive Maps and Entropy Inequalities},
  journal = {Communications in Mathematical Physics},
  volume = {40},
  number = {2},
  pages = {147--151},
  year = {1975},
  doi = {10.1007/BF01609396}
}

@article{Petz86,
  author = {Petz, Denes},
  title = {Sufficient Subalgebras and the Relative Entropy of States of a von Neumann Algebra},
  journal = {Communications in Mathematical Physics},
  volume = {105},
  number = {1},
  pages = {123--131},
  year = {1986},
  doi = {10.1007/BF01212345}
}

@article{Petz88,
  author = {Petz, Denes},
  title = {Sufficiency of Channels over von Neumann Algebras},
  journal = {Quarterly Journal of Mathematics},
  volume = {39},
  number = {1},
  pages = {97--108},
  year = {1988},
  doi = {10.1093/qmath/39.1.97}
}

@article{SFR,
  author = {Sutter, David and Fawzi, Omar and Renner, Renato},
  title = {Universal Recovery Map for Approximate Markov Chains},
  journal = {Proceedings of the Royal Society A},
  volume = {472},
  number = {2186},
  pages = {20150623},
  year = {2016},
  doi = {10.1098/rspa.2015.0623}
}

@article{Wilde,
  author = {Wilde, Mark M.},
  title = {Recoverability in Quantum Information Theory},
  journal = {Proceedings of the Royal Society A},
  volume = {471},
  number = {2182},
  pages = {20150338},
  year = {2015},
  doi = {10.1098/rspa.2015.0338}
}

@article{GaoWilde,
  author = {Gao, L. and Wilde, M. M.},
  title = {Recoverability for Optimized Quantum {$f$}-Divergences},
  journal = {Journal of Physics A: Mathematical and Theoretical},
  volume = {54},
  number = {38},
  pages = {385302},
  year = {2021},
  doi = {10.1088/1751-8121/ac1dc2},
  eprint = {2008.01668},
  archivePrefix = {arXiv},
  url = {https://arxiv.org/abs/2008.01668}
}

@article{MullerLennert2013,
  author = {M{\"u}ller-Lennert, Martin and Dupuis, Fr{\'e}d{\'e}ric and Szehr, Oleg and Fehr, Serge and Tomamichel, Marco},
  title = {On quantum {R{\'e}nyi} entropies: A new generalization and some properties},
  journal = {Journal of Mathematical Physics},
  volume = {54},
  pages = {122203},
  year = {2013},
  doi = {10.1063/1.4838856},
  eprint = {1306.3142},
  archivePrefix = {arXiv}
}

@article{WildeWinterYang2014,
  author = {Wilde, Mark M. and Winter, Andreas and Yang, Dong},
  title = {Strong Converse for the Classical Capacity of Entanglement-Breaking and {Hadamard} Channels via a Sandwiched {R{\'e}nyi} Relative Entropy},
  journal = {Communications in Mathematical Physics},
  volume = {331},
  pages = {593--622},
  year = {2014},
  doi = {10.1007/s00220-014-2122-x},
  eprint = {1306.1586},
  archivePrefix = {arXiv}
}

@article{FangFawzi2021,
  author = {Fang, Kun and Fawzi, Hamza},
  title = {Geometric {R{\'e}nyi} Divergence and its Applications in Quantum Channel Capacities},
  journal = {Communications in Mathematical Physics},
  volume = {384},
  pages = {1615--1677},
  year = {2021},
  doi = {10.1007/s00220-021-04064-4},
  eprint = {1909.05758},
  archivePrefix = {arXiv}
}

@article{BelavkinStaszewski1982,
  author = {Belavkin, V. P. and Staszewski, P.},
  title = {{$C^*$}-algebraic generalization of relative entropy and entropy},
  journal = {Annales de l'Institut Henri Poincar{\'e}, Section A, Physique Th{\'e}orique},
  volume = {37},
  number = {1},
  pages = {51--58},
  year = {1982},
  url = {https://www.numdam.org/item/AIHPA_1982__37_1_51_0/}
}

@article{BluhmCapel2020,
  author = {Bluhm, Andreas and Capel, {\'A}ngela},
  title = {A strengthened data processing inequality for the {Belavkin--Staszewski} relative entropy},
  journal = {Reviews in Mathematical Physics},
  volume = {32},
  number = {2},
  pages = {2050005},
  year = {2020},
  doi = {10.1142/S0129055X20500051},
  eprint = {1904.10768},
  archivePrefix = {arXiv}
}

@article{Schechter1959,
  author = {Schechter, Samuel},
  title = {{On the Inversion of Certain Matrices}},
  journal = {Mathematical Tables and Other Aids to Computation},
  volume = {13},
  number = {66},
  pages = {73--77},
  year = {1959},
  doi = {10.2307/2001955},
  url = {https://doi.org/10.2307/2001955}
}

@misc{GrinbergInverseCauchy,
  author = {Grinberg, Darij},
  title = {The entry sum of the inverse {Cauchy} matrix},
  year = {2026},
  note = {Note, version of June 6, 2026, Theorem 2.1},
  url = {https://www.cip.ifi.lmu.de/~grinberg/algebra/invcauchy.pdf}
}

@article{Alberti1983,
  author = {Alberti, Peter M.},
  title = {A note on the transition probability over {$C^*$}-algebras},
  journal = {Letters in Mathematical Physics},
  volume = {7},
  number = {1},
  pages = {25--32},
  year = {1983},
  doi = {10.1007/BF00398708}
}

@article{HiaiMosonyiPetzBeny2011,
  author = {Hiai, Fumio and Mosonyi, Mil{\'a}n and Petz, D{\'e}nes and B{\'e}ny, C{\'e}dric},
  title = {Quantum {$f$}-divergences and error correction},
  journal = {Reviews in Mathematical Physics},
  volume = {23},
  number = {7},
  pages = {691--747},
  year = {2011},
  doi = {10.1142/S0129055X11004412},
  eprint = {1008.2529},
  archivePrefix = {arXiv}
}

@book{Melnikov2012,
  author = {Melnikov, Yuri A. and Melnikov, Max Y.},
  title = {Green's Functions: Construction and Applications},
  publisher = {De Gruyter},
  address = {Berlin},
  year = {2012},
  doi = {10.1515/9783110253399}
}

@article{BrandaoHarrowOppenheimStrelchuk2015,
 author={Brand{\~a}o, Fernando G. S. L. and Harrow, Aram W. and Oppenheim, Jonathan and Strelchuk, Sergii},
 title={Quantum Conditional Mutual Information, Reconstructed States, and State Redistribution},
 journal={Physical Review Letters}, volume={115}, pages={050501}, year={2015},
 doi={10.1103/PhysRevLett.115.050501}, eprint={1411.4921}, archivePrefix={arXiv}
}

@article{IbinsonLindenWinter2008,
 author={Ibinson, Ben and Linden, Noah and Winter, Andreas},
 title={Robustness of Quantum {Markov} Chains},
 journal={Communications in Mathematical Physics}, volume={277}, pages={289--304}, year={2008},
 doi={10.1007/s00220-007-0362-8}, eprint={quant-ph/0611057}, archivePrefix={arXiv}
}

@article{FaulknerHollandsSwingleWang2022,
 author={Faulkner, Thomas and Hollands, Stefan and Swingle, Brian and Wang, Yixu},
 title={Approximate Recovery and Relative Entropy {I}: General von {Neumann} Subalgebras},
 journal={Communications in Mathematical Physics}, volume={389}, pages={349--397}, year={2022},
 doi={10.1007/s00220-021-04143-6}, eprint={2006.08002}, archivePrefix={arXiv}
}

@article{FaulknerHollands2022,
 author={Faulkner, Thomas and Hollands, Stefan},
 title={Approximate Recoverability and Relative Entropy {II}: 2-positive Channels of General von {Neumann} Algebras},
 journal={Letters in Mathematical Physics}, volume={112}, pages={26}, year={2022},
 doi={10.1007/s11005-022-01510-9}, eprint={2010.05513}, archivePrefix={arXiv}
}

@article{BertaLemmWilde2015,
 author={Berta, Mario and Lemm, Marius and Wilde, Mark M.},
 title={Monotonicity of Quantum Relative Entropy and Recoverability},
 journal={Quantum Information and Computation}, volume={15}, number={15--16},
 pages={1333--1354}, year={2015}, eprint={1412.4067}, archivePrefix={arXiv},
 url={https://arxiv.org/abs/1412.4067}
}

@book{Sutter2018,
 author={Sutter, David},
 title={Approximate Quantum {Markov} Chains},
 series={SpringerBriefs in Mathematical Physics},
 publisher={Springer}, address={Cham}, year={2018},
 doi={10.1007/978-3-319-78732-9},
 url={https://doi.org/10.1007/978-3-319-78732-9},
 eprint={1802.05477}, archivePrefix={arXiv}
}

@article{BluhmCapelCostaRicoJencova2025,
  author = {Bluhm, Andreas and Capel, {\'A}ngela and Costa Rico, Pablo and Jen{\v c}ov{\'a}, Anna},
  title = {{Belavkin--Staszewski} Quantum Markov Chains},
  journal = {Annales Henri Poincar{\'e}},
  year = {2025},
  doi = {10.1007/s00023-025-01644-1},
  eprint = {2501.09708},
  archivePrefix = {arXiv},
  primaryClass = {quant-ph}
}

@misc{rubboli2026logeuclideanrenyiconditionalmutual,
  title = {Log-Euclidean {R\'enyi} Conditional Mutual Information},
  author = {Rubboli, Roberto and Arqand, Amir and Wilde, Mark M.},
  year = {2026},
  eprint = {2609.34699},
  archivePrefix = {arXiv},
  primaryClass = {quant-ph},
  url = {https://arxiv.org/abs/2609.34699}
}

@misc{Chen2026HierarchicalCMI,
  title = {Renormalization group and long-range conditional mutual information in hierarchical models},
  author = {Chen, Yu-Hsueh},
  year = {2026},
  eprint = {2609.02141},
  archivePrefix = {arXiv},
  primaryClass = {quant-ph},
  url = {https://arxiv.org/abs/2609.02141}
}

\end{document}